\documentclass[12pt]{article}
\usepackage[utf8]{inputenc}
\usepackage{amsmath}
\usepackage{graphicx}
\usepackage{amsthm}
\usepackage{natbib}
\usepackage{amssymb}
\usepackage{dsfont}
\usepackage{todonotes}
\usepackage{comment}
\usepackage{multirow}
\usepackage{booktabs}
\usepackage{caption}
\usepackage{pdfpages}
\usepackage{float}
\usepackage{placeins}
\usepackage{bm}
\usepackage{soul}

\numberwithin{equation}{section}
\theoremstyle{definition}

\theoremstyle{remark}

\newtheorem {example}{Example}[section]
\theoremstyle{plain}
\newtheorem{theorem}{Theorem}[section]

\newtheorem{assumption}{Assumption}

\usepackage{hyperref}
\usepackage{cleveref}



\title{Nonparametric Statistical Inference for Multivariate Niche Overlap}

\author{
Jonas Beck\\
\small Department of Artificial Intelligence and Human Interfaces\\
\small Paris Lodron University of Salzburg\\
\small Hellbrunner Straße 34, 5020 Salzburg, Austria
\and
Solomon W. Harrar\\
\small Dr. Bing Zhang Department of Statistics\\
\small University of Kentucky\\
\small 725 Rose Street, Lexington, KY 40536, USA\\
\small Corresponding author: solomon.harrar@uky.edu
}

\date{\today}

\begin{document}
\maketitle

\begin{abstract}
In ecological studies niche overlap is often used to quantify species interaction and dynamics. 
This paper develops a robust, nonparametric statistical framework for quantifying and analyzing multivariate niche overlap. 
Parametric methods are often constrained by restrictive assumptions and tend to underperform in complex multivariate settings. We introduce a nonparametric overlap index and propose  estimators for it. Further, we 
investigate asymptotic properties of the estimators. We also propose bootstrap-based inference procedures that enable statistical testing and simultaneous confidence intervals in small sample settings. Extensive numerical examples demonstrate that our proposed methods maintain correct size and exhibit robust power across various scenarios. We illustrate the practical utility of our methodology using stable isotope measurements from multiple fish species and provide distinct ecological insights regarding species niche differentiation. 
\end{abstract}

\section{Introduction}
Understanding the concept of niche overlap is essential for studying species interactions and ecological dynamics. In ecology, a species’ niche refers to its role and requirements in an ecosystem –- famously described as an "n-dimensional hypervolume" of environmental conditions and resources \cite{hutchinson1957population}.  Niche overlap describes the extent to which two (or more) species use the same resources or environmental conditions, essentially the portion of their niches that they share. 

When two species heavily overlap in resource use, they essentially compete for the same limiting factors. Ecologists often assume that the intensity of competition between species is proportional to their niche overlap. A high degree of niche overlap implies that each species reduces the resources available to the other, lowering each other’s growth or fitness. If niche overlap is complete (i.e., two species share almost identical niches), the classic competitive exclusion principle predicts that they cannot coexist stably \cite{schoener, chesson}. 
Beyond their variety of applications in ecology, the concept of niche overlap is also used in fields such as economics and marketing \citep{milne, dimmick}, as well as sociology and human geography \citep{Audia2006OrganizationalFI, freeman, hannan}.

To model ecological niches, we use different distributions for each species. To quantify niche overlap, we assess the intersection of the corresponding distribution functions.
Despite its importance, traditional methods have focused predominantly on overly simplistic, not always justifiable parametric models \cite{swanson_15, parra} or do not allow any statistical inference \cite{juncker, blonder2018new} or are restricted to univariate cases \citep{langthaler}, limiting their applicability in complex real-world scenarios where multiple factors interact simultaneously.

This paper extends the niche overlap framework developed by \citet{ParkinsonKutilKupplerJunkerTrutschnigBathke+2018} and \citet{langthaler} to the multivariate domain, enabling a more comprehensive analysis of ecological data.
We propose a robust  statistical method for estimating multivariate niche overlap. More specifically, we introduce a nonparametric overlap index, provide a consistent estimation, investigate its asymptotic properties, and develop resampling-based approaches for hypothesis testing and confidence interval construction. This methodological advancement provides a solid foundation for analyzing more complex ecological settings involving multiple species and multivariate data, which we discuss in subsequent sections.

The remainder of the paper is organized as follows: in Section \ref{sec:3} we present the main theoretical framework for the multi-sample setting. The special case of two samples is provided in the Supplement (Sections \ref{sec:2}–\ref{s3}). We present extensive simulations to assess the numerical  performance and robustness of our proposed methods in Section \ref{sim}.  In Section \ref{casestudy}, we illustrate our approach with a real-world ecological dataset, demonstrating its effectiveness in capturing niche overlap in a multivariate setting.  Finally, we conclude the paper with some  discussion in Section \ref{discussion}.
\subsection{Motivating Example} \label{mot_ex}
We motivate our method with a dataset collected by \citet{swanson_15}. The dataset involves stable isotope measurements from muscle tissues of four fish species: Lake Whitefish (Coregonus clupeaformis, LKWF) ($n_1=69$), Broad Whitefish (Coregonus nasus\, BDWF) ($n_2=71$), Arctic Cisco (Coregonus autumnalis, ARCS) ($n_3=67$) and Least Cisco (Coregonus sardinella, LSCS)($n_4=70$). Here $n_i$ refers to the sample size in each group. These fishes were sampled from Phillips Bay, Beaufort Sea, Canada. Each individual fish was analyzed for three stable isotope ratios: Carbon isotope ratio ($\delta^{13}\mathrm{C}$, D13C), Nitrogen isotope ratio ($\delta^{15}\mathrm{N}$, D15N) and  Sulfur isotope ratio ($\delta^{34}\mathrm{S}$, D34S). A boxplot of the data is shown in Figure \ref{fig:boxplot_fish}. A major question of interest here is which species share similar dietary and habitat characteristics. Quantifying niche overlap among these species is therefore essential to assess the extent to which they exploit similar resources and habitats.

\begin{figure}[h]
    \centering
    \includegraphics[width=\textwidth]{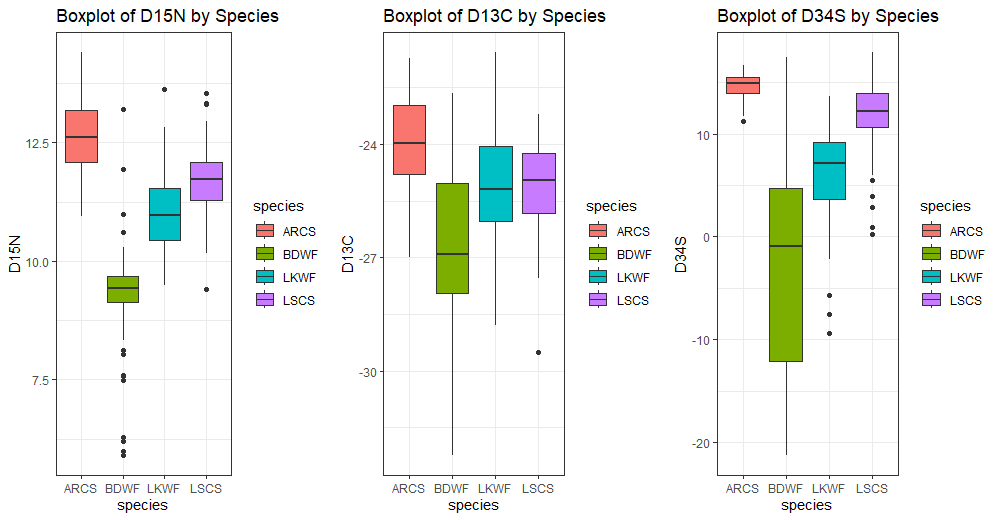}
    \caption{Boxplots of D15N, D13C, and D34S by Species}
    \label{fig:boxplot_fish}
\end{figure}
The authors define ecological niches probabilistically, assuming that the stable isotope measurements for each species follow a multivariate normal distribution with different parameters. They employed Bayesian methods to estimate these parameters and derive posterior distributions to represent statistical uncertainty. However, the assumption of multivariate normality does not seem justifiable. For example, the Henze-Zirkler Test \cite{Henze1990ACO} shows highly significant p-values of $p < 0.0001$. The Q-Q plots in Figure \ref{fig:qqplot_fish} in the Supplement confirm  this finding. Additional QQ-plots and exploratory analyses are provided in the Supplement (Section \ref{s4}).

\section{Statistical Methods} \label{sec:3}
We now introduce the general statistical framework for quantifying niche overlap 
in a multivariate setting involving multiple species or populations. 
Our aim is to define and estimate suitable indices of niche overlap based on these distributions 
and to develop an accompanying inference framework.

\subsection{Statistical Model}

\label{sec:setting}
For the remainder of this section, let \( \mathbf{F}_1, \ldots, \mathbf{F}_k \) be absolutely continuous, non-degenerate multivariate cumulative distribution functions on \(\mathbb{R}^d\). Let
\[
\mathbf{X}_{i1}, \mathbf{X}_{i2}, \ldots, \mathbf{X}_{in_i} \overset{\rm i.i.d.}{\sim} \mathbf{F}_i,  
\]
for $i = 1, \ldots, k$
be independent samples from these distribution functions, where
$
\mathbf{X}_{ik} = (X_{ik}^{(1)}, \ldots, X_{ik}^{(d)})^\top . 
$
The marginal distributions of $ \mathbf{F}_i$  are denoted by $F_i^{(1)}, \ldots, F_i^{(d)}$.
We define \( N = \sum_{i = 1}^{k} n_i \). Define the reference distribution 
\[ \mathbf{H}(t) = \sum_{i=1}^{k} \lambda_i \mathbf{F}_i(t), \]
where \( \sum_{i = 1}^{k} \lambda_i = 1 \) and \( \lambda_i \geq 0, \, i = 1, \ldots, k \) \cite{brunner17}. The marginal distributions of $\mathbf{H}$ are denoted by $H^{(1)}$, \ldots, $H^{(d)}$.

We define the niche overlap of \( \mathbf{F}_i \) with respect to \( \mathbf{H} \) as
\begin{align*}
    \mathcal{I}(\mathbf{H},\mathbf{F}_i)= (I(H^{(1)},F_i^{(1)}), \ldots, I(H^{(d)},F_i^{(d)})),
\end{align*} 
where 
\begin{align*}
      I(H^{(s)},F_i^{(s)}) = P(X_{i,\text{low}}^{(s)} < Z^{(s)} < X_{i,\text{up}}^{(s)}), 
\end{align*}
$X_{i,\text{low}}^{(s)}\sim F_{i, {\rm low}}^{(s)}$ independently of $X_{i,\text{up}}^{(s)}\sim F_{i,{\rm up}}^{(s)}$, $F_{i, {\rm low}}^{(s)}$ is defined as the conditional distribution of $X_i^{(s)}$ on being smaller than its own median and $F_{i, {\rm up}}^{(s)}$ as the conditional distribution on being larger than its own median. The variable $Z^{(s)}$ is drawn from the reference distribution function $H^{(s)}$.

Thus, for each component $s$, the index $I\!\bigl(H^{(s)},F_i^{(s)}\bigr)$ is large when $F_i^{(s)}$ places most of its mass around the central region of $H^{(s)}$ (near its median), and small when $F_i^{(s)}$ is concentrated far below or above the median of $H^{(s)}$. Equivalently, it is the probability that a draw $Z^{(s)}\sim H^{(s)}$ falls between two independent draws from $F_i^{(s)}$—one conditioned to lie below its median and one above.

The quantity $I(H,F_i)$ is defined as a $d$-dimensional vector of marginal
overlap indices $I(H^{(s)},F_i^{(s)})$, $s=1,\ldots,d$.
Thus, the parameter is intentionally component-wise and does not aim to
summarize similarity of the full joint distributions by a single scalar
overlap measure. The advantage of this construction lies in its
interpretability in terms of individual niche dimensions, for example, to
identify which variables primarily drive niche overlap or differentiation.
The multivariate structure enters through joint statistical inference for the
$kd$-dimensional parameter vector, where dependence across components is
reflected in the asymptotic covariance matrix $\Sigma$ in Theorem~2.1.
Hypothesis based on marginal distributions naturally arise in some multivariate inference problems.  The most commonly occurring example is inference about mean vectors when higher moments or joint distributions of the groups are not necessarily the same (\cite{anderson2003introduction}, p. 187). In  multivariate analysis this problem is referred to as the Behrens-Fisher Problem.  See also \cite{brunner_munzel} and \cite{brunner2002multivariate}  for a nonparametric version.

The overlap index can be rewritten  \citep{ParkinsonKutilKupplerJunkerTrutschnigBathke+2018, langthaler} as
\begin{align*}
    I(H^{(s)},F_i^{(s)}) &= \int_0^1 H^{(s)} \circ (F_i^{(s)})^{-1}(1 - \alpha/2) \, d\alpha - \int_0^1 H^{(s)} \circ (F_i^{(s)})^{-1}(\alpha/2) \, d\alpha \nonumber \\
    &= 2 \left( \int_{(F_i^{(s)})^{-1}(1/2)}^{\infty} H^{(s)} \, dF^{(s)}_i - \int^{(F^{(s)}_i)^{-1}(1/2)}_{-\infty} H^{(s)} \, dF^{(s)}_i \right)
\end{align*}
for $s=1, \ldots, d$. The first term equation here shows the close relation to the receiver operating characteristic curve (ROC-curve), given by $H^{(s)} \circ \bigl(F_i^{(s)}\bigr)^{-1}(u)$.

Here, we focus on the weighted niche overlap, by defining $\lambda_i:= n_i/N$. Of course, other weighting schemes would be possible. For example the unweighted niche overlap $\lambda_i=1/k$, which leads to another estimator \cite{brunnerbathkekonnietschke}. The weights $\lambda_i$ define the reference distribution 
$H = \sum_{i=1}^k \lambda_i F_i$ and are therefore part of the target parameter 
$I(H,F_i)$. The proposed estimation and inference framework remains valid for any
given weighting scheme $\lambda$. Different weighting schemes correspond
to different scientific questions: while $\lambda_i = 1/k$ yields a species-level
reference that can be useful when all groups are to be treated symmetrically,
$\lambda_i = n_i/N$ defines a pooled, individual-level reference that reflects the
composition of the observed community. This interpretation is particularly
appropriate when the sampling design is intended to approximate the underlying
population structure (e.g., relative abundance or availability). Hence, differences
across weighting schemes should be interpreted as differences in the underlying
reference population rather than as methodological discrepancies.

The framework can also be formulated in the classical two-sample setting, which 
serves as a special case of the above. A brief summary is given here, while full details 
(statistical model, estimation, asymptotics, resampling, and simulations) are deferred 
to the Supplement (Sections \ref{sec:2}–\ref{s3}).

Due to the linearity of \( \mathcal{I}(\cdot, \cdot) \) in its first argument, we can write the \( k \)-sample niche overlap as a convex combination of all pairwise overlaps with \( \mathbf{F}_i \), i.e.,
\[
\mathcal{I}(\mathbf{H},\mathbf{F}_i) = \sum_{j=1}^k \lambda_j \mathcal{I}(\mathbf{F}_j,\mathbf{F}_i),
\]
where the pairwise overlap of \( \mathbf{F}_i \) with respect to \( \mathbf{F}_j \) is defined as
\begin{align*}
    \mathcal{I}(\mathbf{F}_j,\mathbf{F}_i)= (I(F_j^{(1)},F_i^{(1)}), \ldots, I(F_j^{(d)},F_i^{(d)})),
\end{align*}
with
\begin{align*}
      I(F_j^{(s)},F_i^{(s)}) = P(X_{i,\text{low}}^{(s)} < X_j^{(s)} < X_{i,\text{up}}^{(s)}). 
\end{align*}
This equality would not hold if we switch the role of $\mathbf{H}$ and $\mathbf{F}_i$.

For more than two samples, one may consider the pairwise niche overlap. However, these do not correspond to any notion of \emph{overall} overlap between one distribution and the other distributions. In the context of hypothesis testing, this  would require multiplicity adjustments for \( k \times (k-1) \) hypothesis tests, instead of only \( k \), where $k$ is the number of samples.  In the literature on nonparametric effects, the approach of comparing all \( k \) distributions to the same reference distribution has been very effective and also avoids paradoxical results due to nontransitivity \citep{Brown2002KruskalWallisMC, Thangavelu2007WilcoxonMannWhitneyTF, brunner17}. We therefore adopt this approach for the multivariate niche overlap.

Since the two-sample setting is a special case of the general framework, we focus on the multi-sample case in the following. Full details for the two-sample theory (definitions, asymptotics, resampling, and simulations) are provided in the Supplement (Sections \ref{sec:2}–\ref{s3}).

Similar as in the univariate case not all values in the unit cube $ [0,1]^d$ are attainable. As can be easily seen the set of possible values for the d-dimensional niche overlap is  
$
    \prod_{s=1}^d \Bigl [\frac{1}{2 n_i}, 1- \frac{1}{2 n_i} \Bigl ].$ 

Note that $F_j^{(s)}=F_i^{(s)}$ for all $i,j$ implies $I\!\bigl(H^{(s)},F_i^{(s)}\bigr)=\tfrac12$, while the converse need not hold. 
Further, for any two distributions,
\[
0 \;\le\; I\!\bigl(H^{(s)},F_i^{(s)}\bigr)+I\!\bigl(F_i^{(s)},H^{(s)}\bigr) \;\le\; 1,
\]
and the upper bound is attained when the medians of $H^{(s)}$ and $F_i^{(s)}$ coincide \citep[][Lemma~2.10]{ParkinsonKutilKupplerJunkerTrutschnigBathke+2018}.

\subsection{Estimation} \label{sec:estmore}
We define the empirical versions of the marginal cumulative distribution functions as
\begin{align*}
     \hat{F}^{(s)}_{i,n_i}(t) := \frac{1}{n_i} \sum_{r = 1}^{n_i} \mathbf{1}_{(-\infty,t]}(X_{i,r}^{(s)})\quad \text{ and}\quad \hat{F}_{i, n_i}(t) = (\hat{F}^{(1)}_{i,n_i}(t), \ldots, \hat{F}_{i,n_i}^{(d)}(t))^\top. 
     \end{align*}
     The empirical version of the reference distribution is defined by
     \begin{align*}
     \hat{H}_{N}(t) = \sum_{i = 1}^{k} \lambda_i \hat{F}_{i,n_i}(t).
\end{align*}

Denote by  \( R_{i,r}^{(s)} \)  the rank of \( X_{i,r}^{(s)} \) in the combined samples for the $s$-th component, and \( R_{i, r}^{(i)(s)} \) the rank of \( X_{i, r}^{(s)} \) in the \( i \)-th sample for the $s$-th component, i.e.,
\begin{align*}
    R_{i,r}^{(s)} = \sum_{j = 1}^{k} \sum_{r' = 1}^{n_j} \mathbf{1}_{(-\infty,X_{i,r}^{(s)}]}(X_{j,r'}^{(s)}) \quad \text{and} \quad R_{i,r}^{(i)(s)} := \sum_{r' = 1}^{n_i} \mathbf{1}_{(-\infty,X_{i,r}^{(s)}]}(X_{i, r'}^{(s)}).
\end{align*}

We estimate the multi-sample indices by plugging empirical marginals into the definition
of $I(H^{(s)},F_i^{(s)})$. The plug-in estimator for group $i$ is
\begin{align*}
    \widehat{\mathcal{I}}_{N,n_i}(\mathbf{H},\mathbf{F}_i) = (I(\hat{H}_{N}^{(1)}, \hat{F}_{i, n_i}^{(1)}), \ldots, I(\hat{H}_{N}^{(d)}, \hat{F}_{i, n_i}^{(d)}))^T,
\end{align*}
with componentwise form
\begin{align*}
I\bigl(\hat H^{(s)}_N,\hat F^{(s)}_{i,n_i}\bigr)
&=\int_0^1 \hat H^{(s)}_N\!\circ\!(\hat F^{(s)}_{i,n_i})^{-1}(1-\alpha/2)\,d\alpha
   \;-\;\int_0^1 \hat H^{(s)}_N\!\circ\!(\hat F^{(s)}_{i,n_i})^{-1}(\alpha/2)\,d\alpha\\
&=2\!\left(\int_{(\hat F^{(s)}_{i,n_i})^{-1}(1/2)}^{\infty} \hat H^{(s)}_N\,d\hat F^{(s)}_{i,n_i}
          \;-\;\int_{-\infty}^{(\hat F^{(s)}_{i,n_i})^{-1}(1/2)} \hat H^{(s)}_N\,d \hat F^{(s)}_{i,n_i} \right ),
          \end{align*}
          where $(\widehat F^{(s)}_{i,n_i})^{-1}(u):=\inf\{y:\widehat F^{(s)}_{i,n_i}(y)\ge u\}$ is the empirical quantile.
By the componentwise consistency of the two-sample estimators and because $d$ and $k$ are finite,
$\widehat{\mathcal I}_{N,n_i}(H,F_i)$ is strongly consistent for $\mathcal I(H,F_i)$ for each $i=1,\ldots,k$.

To express these estimators in terms of ranks we reorder the observations within each sample such that $X_{i,1}^{(s)} \leq X_{i,2}^{(s)} \leq \ldots \leq X_{i,n_i}^{(s)},\, i = 1,\ldots,k$ and define $l_i =\lceil n_i/2 \rceil $. 
Using ranks
we can rewrite 
\begin{align*}
   & I(\hat{H}^{(s)}_{N},\hat{F}^{(s)}_{i,n_i}) 
    = \frac{2}{n_i N} \Bigl (\sum_{r=l_i+1}^{n_i} R_{i,r}^{(s)} - \sum_{r = l_i +1}^{n_i} R_{i,r}^{(i)(s)} - \sum_{r=1}^{c_i} R_{i,r}^{(s)} + \sum_{r=1}^{c_i} R_{i,r}^{(i)(s)} \Bigl ),
\end{align*}
where $c_i=l_i$ when $n_i$ is even and $c_i=l_i-1$ when  $n_i$ is odd. For the derivation we refer to \ref{A3}.

\subsection{Asymptotic Distribution and Resampling}
In this subsection we exploit the analogy to ROC curves, given by
$H^{(s)} \circ \bigl(F_i^{(s)}\bigr)^{-1}(u)$, to derive the asymptotic distribution of the multi-sample overlap estimator. This analogy is used for distributional comparison and not in the sense of classical ROC
classification theory.

To formulate the main theorem, we first need the following assumptions.
\begin{assumption} \label{assum_3}
The ratio of the sample sizes \( N/n_i \rightarrow \tau_i \in (1,\infty)\) for each \(i = 1, \ldots, k \).  
\end{assumption}
The slope (derivative) of the ROC curve $H^{(s)} \circ \bigl(F_i^{(s)}\bigr)^{-1}(u)$ is
\[
\frac{d}{du}\Bigl[H^{(s)} \circ \bigl(F_i^{(s)}\bigr)^{-1}(u)\Bigr]
= \frac{h^{(s)}\!\bigl((F_i^{(s)})^{-1}(u)\bigr)}{f_i^{(s)}\!\bigl((F_i^{(s)})^{-1}(u)\bigr)},
\]
where $h^{(s)}$ and $f_i^{(s)}$ denote the corresponding densities.
\begin{assumption} \label{assum_4}
The slope of the ROC, $H^{(s)} \circ {F_i^{(s)}}^{-1}(u)$, is bounded on  the interval $(a, b)$  for any $0<a < b<1$ and any $s=1,\ldots,d$.
\end{assumption}
Assumption~\ref{assum_3} is a standard proportional-growth condition ensuring comparable
contributions from each group. Assumption~\ref{assum_4} is equivalent to requiring that the
ratio $h^{(s)}((F_i^{(s)})^{-1}(u))/f_i^{(s)}((F_i^{(s)})^{-1}(u))$ be finite for all
$u\in(0,1)$. Under Assumptions~\ref{assum_3}–\ref{assum_4}, the plug-in estimator of the ROC
curve is strongly consistent \citep{hsieh_turnbull}.
Before we can state the main theorem we have to define some notation regarding the bootstrap: 

For each group $i=1,\dots,k$, draw a bootstrap sample of $d$-dimensional vectors
$X_{i,1}^*,\dots,X_{i,n_i}^*$ i.i.d.\ from the empirical distribution
$\widehat F_{i,n_i}$ on $\mathbb{R}^d$, i.e.\ by sampling with replacement from the observed
vectors $\{X_{i,1},\dots,X_{i,n_i}\}$.
For each component $s=1,\dots,d$, define the corresponding marginal empirical CDF from
the resampled vectors by
\[
\widehat F^{(s)*}_{i,n_i}(t)
:= \frac{1}{n_i}\sum_{r=1}^{n_i}\mathbf 1\{X^{(s)*}_{i,r}\le t\},
\qquad
\widehat H_N^{(s)*}(t)
:=\sum_{j=1}^k \lambda_j\,\widehat F^{(s)*}_{j,n_j}(t).
\]
The bootstrap plug-in estimator for group $i$ is then
\[
\widehat{\mathcal I}^{*}_{N,n_i}(H,F_i)
=
\bigl(
I(\widehat H_N^{(1)*},\widehat F^{(1)*}_{i,n_i}),
\dots,
I(\widehat H_N^{(d)*},\widehat F^{(d)*}_{i,n_i})
\bigr)^\top,
\]
and stacking over $i$ yields $\widehat{\mathcal I}^{*}_{N}(H,F_{1,\dots,k})$ analogously.

For brevity, we define:
\begin{align*}
 &\mathcal{I}(\mathbf{H},\mathbf{F}_{1,\ldots,k}) := \left(   \mathcal{I}(\mathbf{H},\mathbf{F}_1), \ldots , \mathcal{I}(\mathbf{H},\mathbf{F}_k) \right)^T,  \\
&\widehat{\mathcal{I}}_N(\mathbf{H},\mathbf{F}_{1,\ldots,k}) := \left( \widehat{\mathcal{I}}_{N,n_1}(\mathbf{H},\mathbf{F}_1), \ldots , \widehat{\mathcal{I}}_{N,n_k}(\mathbf{H},\mathbf{F}_k) \right)^T \text{ and } \\
&  \widehat{\mathcal{I}}^*_N(\mathbf{H},\mathbf{F}_{1,\ldots,k}) := \left( \widehat{\mathcal{I}}^*_{N,n_1}(\mathbf{H},\mathbf{F}_1), \ldots , \widehat{\mathcal{I}}^*_{N,n_k}(\mathbf{H},\mathbf{F}_k) \right)^T .
\end{align*}

We now state the main theorem.
\begin{theorem}
    \label{thm:asym:2}
  Under Assumptions \ref{assum_3} and \ref{assum_4} it follows that the vector 
     \begin{align}\label{eq:conv,mult}
         \sqrt{N} \biggl ( \widehat{\mathcal{I}}_N(\mathbf{H},\mathbf{F}_{1,\ldots,k}) -\mathcal{I}(\mathbf{H},\mathbf{F}_{1,\ldots,k}) \biggl )
     \end{align}
   is asymptotically normal distributed with mean vector $\textbf{0}_{kd} \in \mathbb{R}^{kd} $ and nonnegative definite covariance matrix $\Sigma$. 
   Additionally, the bootstrapped version 
  
 \begin{align} \label{eq: boot,mult}
     \sqrt{N} \biggl ( \widehat{\mathcal{I}}^*_N(\mathbf{H},\mathbf{F}_{1,\ldots,k}) - \widehat{\mathcal{I}}_N(\mathbf{H},\mathbf{F}_{1,\ldots,k}) \biggl )
 \end{align}
   converges in outer probability conditionally on the data to the same limit distribution as that of \eqref{eq:conv,mult}.
   \begin{proof}
See \ref{A4}
 \end{proof}
\end{theorem}
Although Theorem \ref{thm:asym:2} establishes asymptotic normality with covariance matrix $\Sigma$,
closed-form expressions for the entries of $\Sigma$ (obtainable via covariances of Brownian bridges arising in the ROC-type mapping) are algebraically cumbersome and of limited
practical value. Accordingly, in all implementations we approximate the sampling
distribution and estimate $\Sigma$ by a nonparametric bootstrap of the empirical distribution on $\mathbb{R}^{d}$,
using the resulting covariance estimate $\widehat{\Sigma}^{\,*}$ (or bootstrap quantiles) for
test statistics and simultaneous confidence intervals.

\subsection{Hypothesis Testing and Confidence Intervals} \label{sec3.4}
In this section we will use the asymptotic results in  Theorem \ref{thm:asym:2} for statistical inference.
As previously stated, for equal distributions, $\mathbf{H}=\mathbf{F}_i$ for each $i=1, \ldots,k$, we have \( I(\mathbf{H}, \mathbf{F}_i)=\frac{1}{2}\mathbf{1}_{d} \). Therefore, here we also use \( \frac{1}{2} \) as a benchmark value for testing equality of overlap index across the k-groups. We formulate the hypotheses as  
\begin{align*}
H_0 : \mathcal{I}(\mathbf{H},\mathbf{F}_{1,\ldots,k}) =\frac{1}{2} \mathbf{1}_{dk}  \text{ versus }  H_1: \mathcal{I}(\mathbf{H},\mathbf{F}_{1,\ldots,k}) \neq \frac{1}{2} \mathbf{1}_{dk} .
\end{align*}
By Theorem~\ref{thm:asym:2} and the Continuous Mapping Theorem, under
$H_0:\ \mathcal I(H,F_{1,\ldots,k})=\tfrac12\,\mathbf{1}_{dk}$ and under the assumption of a nonsingular covariance matrix the Wald-type quadratic form
     \[
Q_N = N \cdot\ (\widehat{\mathcal{I}}_N(\mathbf{H},\mathbf{F}_{1,\ldots,k})- \frac{1}{2} \mathbf{1}_{dk})^T \Sigma^{-1} (\widehat{\mathcal{I}}_N(\mathbf{H},\mathbf{F}_{1,\ldots,k})-\frac{1}{2} \mathbf{1}_{dk}),
\]
converges in distribution to $\chi^2_{dk}$. Since the bootstrap in Theorem~\ref{thm:asym:2}
is valid, we may replace $\Sigma$ by the bootstrap covariance $\widehat\Sigma^{\,*}$ to obtain the
feasible statistic
      \[
\hat{Q}_N = N \cdot\ (\widehat{\mathcal{I}}_N(\mathbf{H},\mathbf{F}_{1,\ldots,k})- \frac{1}{2} \mathbf{1}_{dk})^T \Sigma^{*^{-1}} (\widehat{\mathcal{I}}_N(\mathbf{H},\mathbf{F}_{1,\ldots,k})-\frac{1}{2} \mathbf{1}_{dk}),
\]
which has the same $\chi^2_{dk}$ limiting distribution. If $\Sigma$ is singular, one may replace $\Sigma^{-1}$ by the Moore--Penrose generalized inverse $\Sigma^{+}$,
leading to a Wald-type statistic based on the effective rank.

It is well documented that the Wald approximation can be liberal in small samples
\citep{cui2021nonparametric,brunner2002multivariate}; our
simulations confirm this. The instability mainly stems from inverting the estimated
covariance matrix in the quadratic form. As a more robust alternative, we adopt an
ANOVA-type statistic (ATS) introduced for parametric models \citep{box} and extended to
univariate and multivariate nonparametric settings
\citep{brunner97,brunner2002multivariate}. The ATS replaces the matrix inverse by the
trace of the covariance estimate, yielding a ratio-of-quadratics similar to classical ANOVA.
Under $H_0$, the statistic
\begin{align*}
F_n &= \frac{N}{\operatorname{tr}(\hat{\Sigma^*})} (\widehat{\mathcal{I}}_N(\mathbf{H},\mathbf{F}_{1,\ldots,k}) - \frac{1}{2} \mathbf{1}_{dk})^{\top} (\widehat{\mathcal{I}}_N(\mathbf{H},\mathbf{F}_{1,\ldots,k}) - \frac{1}{2} \mathbf{1}_{dk}) \\
&= \frac{N}{\operatorname{tr}(\hat{\Sigma^*})} \sum_{l=1}^{d} \sum_{i=1}^{k} \left( \widehat{\mathcal{I}}(H^{(l)},F_i^{(l)}) - \frac{1}{2} \right)^2.
\end{align*}
has an approximate central $F(\hat\nu,\infty)$ distribution with
$\hat\nu=\bigl[\operatorname{tr}(\widehat\Sigma^{\,*})\bigr]^2/
\operatorname{tr}\!\bigl((\widehat\Sigma^{\,*})^{2}\bigr)$; see \citet{brunner2002multivariate}
for details of this approximation. In our simulations, the ATS maintains the nominal level
substantially better than the Wald test in small samples, while exhibiting comparable power.

We now construct confidence regions and simultaneous intervals for $\widehat{\mathcal{I}}_N(\mathbf{H},\mathbf{F}_{1,\ldots,k})$. Since $\widehat{\mathcal{I}}_N(\mathbf{H},\mathbf{F}_{1,\ldots,k})$ is asymptotically normal with covariance
$\Sigma/N$ and the bootstrap is valid (Theorem~\ref{thm:asym:2}),
we get asymptotic \( (1 - \alpha) \)- Simultaneous Confidence Intervals (SCIs) for the niche overlap using
\[
\widehat{\mathcal{I}}_N(\mathbf{H},\mathbf{F}_{1,\ldots,k}) \pm z(1 - \alpha, \Sigma^*)  ,
\]
where $z(1 - \alpha, \Sigma^*)$ is the equi-coordinate multivariate normal quantile which can be  computed using the \textsc{R}-package \textsc{mvtnorm} \citep{mvtnorm, genz_bretz} and the methods described therein.

In small samples, the normal approximation may be inaccurate. A simple alternative uses
coordinatewise bootstrap quantiles of $\sqrt{N}\,(\widehat{\mathcal I}^{\,*}_N(\mathbf{H},\mathbf{F}_{1,\ldots,k})-\widehat{\mathcal I}_N(\mathbf{H},\mathbf{F}_{1,\ldots,k}))$
together with a Bonferroni correction:
 \begin{align*}
 \left[
\widehat{\mathcal{I}}_N(\mathbf{H},\mathbf{F}_{1,\ldots,k})- \frac{1}{\sqrt{N}}    z^*_{*,\alpha/(2kd)},  \widehat{\mathcal{I}}_N(\mathbf{H},\mathbf{F}_{1,\ldots,k})- \frac{1}{\sqrt{N}}    z^*_{*,1-\alpha/(2kd)}
\right]
\end{align*}
where $z^*_{*,1-\alpha/(2kd)}=z^*_{I(\hat{H}, \hat{F}_{1,\ldots,k}), 1-\alpha/(2kd)}$ denotes the empirical bootstrap quantile for each component of the estimator. The Bonferroni-based procedure is conservative, as reflected in our simulations, but is included as a robust finite-sample safeguard against inaccuracies of the asymptotic joint normal approximation.

Furthermore, by the usual duality between hypothesis tests and confidence sets,
inverting the Wald-type test (equivalently, centering it at an arbitrary target vector $\bm v_0$)
yields the elliptical $(1-\alpha)$ confidence region by considering all points
$\bm v=(v_1,\ldots,v_{kd})\in\mathbb{R}^{kd}$ satisfying

\begin{align*}
   \frac{1}{\sqrt{N}} \biggl (\widehat{\mathcal{I}}_N(\mathbf{H},\mathbf{F}_{1,\ldots,k})- {\bm v} \biggl )^T  \Sigma^{*^{-1}}_{N} \biggl( \widehat{\mathcal{I}}_N(\mathbf{H},\mathbf{F}_{1,\ldots,k})- {\bm v} \biggl ) \leq q_{\chi^2_{kd}}(1-\alpha), 
    \end{align*}
    where $q_{\chi^2_{kd}}(1-\alpha)$ is the $1-\alpha $ quantile of a $\chi^2$ distribution with $k  d$ degrees of freedom and  $\Sigma^*_{N}$ is the sample covariance matrix of the bootstrap estimates.  The projection of this region on the individual coordinates gives a simultaneous confidence interval  which we will refer to as elliptical-based confidence intervals. 

\begin{comment}
Due to the asymptotic normality demonstrated in Theorem \ref{thm:asym:2}, we can construct an elliptical \(1-\alpha\) confidence region by considering all \(d\)-dimensional vectors \( {\bf v}=(v_1, \ldots, v_d) \) that satisfy
\[
\left\| 
\left( 
\frac{I(H^{(1)}, F_i^{(1)})}{\sqrt{n}}, \ldots, \frac{I(H^{(d)}, F_i^{(d)})}{\sqrt{n}} 
\right)^T 
- 
\left( 
\frac{v^{(1)}}{\sqrt{n}}, \ldots, \frac{v^{(d)}}{\sqrt{n}} 
\right)^T 
\right\|_2 
\leq 
\Phi_{\Sigma^*}(1 - \alpha)
\]

where \(\Phi_{\Sigma^*}\) denotes the cumulative distribution function of the multivariate normal distribution with mean 0 and covariance matrix \(\Sigma^*\). \(\Sigma^*\) is again the covariance matrix of the bootstrapped sample.

Similar as in Section \ref{sec2.4} we consider again the Wald type statistic with the bootstrapped covariance matrix. We define the Wald-type test statistic as:
  \[
\hat{Q}_{n,i} = n-n_i \cdot\ (\hat{I}(H,F_i)- \frac{1}{2} 1_d)^T \Sigma^{*^{-1}} (\hat{I}(H,F_i)-\frac{1}{2} 1_d).
\]
Therefore this test statistic $\hat{Q}_{n,i}$ is again central $\chi^2$-distributed with $d$ degrees of freedom, due to similar arguments as before.
\end{comment}

Our deduced inference procedure is general and, in particular, it allows to conduct post hoc analysis in a manner similar to \citet{dobler}. When the global null hypothesis $H_0$ of no effect is rejected, it may be of interest to test more specific null hypotheses to find out where the difference in the overlap indices lie. One way to do this is to first test the univariate null hypothesis 
\begin{align*}
    H_0^{(s)} \ : \ I(H^{(s)},F_{1,\ldots,k}^{(s)}) = \frac{1}{2} \mathbf{1}_k
\end{align*}
for each $s=1, \ldots, d$ to detect the univariate endpoint which potentially caused the rejection. Next, one would proceed to test which groups show significant difference
 \begin{align*}
    H_{0,i} \ : \ \mathcal{I}(\mathbf{H},\mathbf{F}_{i}) = \frac{1}{2} \mathbf{1}_d
\end{align*}
for all $i \in 1 ,\ldots, k$. Our results in Theorem \ref{thm:asym:2} can be easily extended to null hypotheses involving contrast matrices by applying the closed testing principle \citep{marcus76}. Note that this approach is not feasible if the global null assumes equality of distributions, as the equality of marginals does not imply equality of joint distributions. However, we leave this extension for future research.

\section{Numerical Examples} \label{sim}

We now present simulation results to assess the performance of our methods. For completeness, additional results for the two-sample case are provided in the Supplement (Section \ref{s2}).
In the following we evaluate the performance of three tests from Section \ref{sec3.4}: (i) the Wald type test (Wald), (ii) the ANOVA-type test (ANOVA), and (iii) the percentile test based on empirical bootstrap quantiles with Bonferroni correction (Percentile). 
For all numerical experiments, bootstrap-based inference was carried out using $B = 2000$ bootstrap resamples.
\begin{example} \label{ex7}
{\bf (Empirical Size)} We assess the empirical size using three multivariate normal distributions ($k=3$) with common mean vector \( \bm{\mu}_{i} = \mathbf{1}\in \mathbb{R}^d \) and  constant covariance matrix \( \Sigma \), where the diagonal entries are equal to $1$ and off-diagonal entries equal to $0.25$. Simulations were performed for dimensions \( d = 2, 3, 4, 5 \) and sample sizes \( n = 50, 100 \).

The Wald and ANOVA tests achieve empirical sizes close to the nominal level, but show noticeable conservatism in some settings, particularly for small samples and higher
dimensions (Table \ref{tab:rejection_rates2}).The strong conservatism of the Percentile test is expected, as it combines coordinatewise
bootstrap quantiles with a Bonferroni correction over $kd$ components, which ignores the
substantial dependence between overlap estimates.
\end{example}

\begin{table}[htb]
    \centering
    \begin{tabular}{|c|c|c|c|c|}
        \hline
        Sample Size & $d$ & Wald Test & ANOVA Test & Percentile Test \\
        \hline
        50  & 2 & 0.023 & 0.029 & 0.003833 \\
        50  & 3 & 0.028 & 0.030 & 0.002222 \\
        50  & 4 & 0.035 & 0.024 & 0.002500 \\
        50  & 5 & 0.029 & 0.018 & 0.001133 \\
        \hline
        100 & 2 & 0.030 & 0.028 & 0.007167 \\
        100 & 3 & 0.036 & 0.041 & 0.005111 \\
        100 & 4 & 0.033 & 0.025 & 0.003000 \\
        100 & 5 & 0.039 & 0.033 & 0.003200 \\
        \hline
    \end{tabular}
   \caption{Empirical sizes for the tests in Section \ref{sec3.4} for \( k=3 \) groups with equal mean and shared covariance matrix. See Example \ref{ex7} for details.}
\label{tab:rejection_rates2}
\end{table}
\begin{example} \label{ex8}
{\bf (Power under Variance Heterogeneity)} We compare three multivariate normal distributions with mean vector \( \bm{\mu}_i = \mathbf{1}\in \mathbb{R}^d \) and  covariance matrices $\Sigma_i$ whose  $(j,k)$th entry $\Sigma_{i,jk}$  defined by \( \Sigma_{i,jk} = \sigma_i \) if \( j = k \), and 0.25 otherwise, with \( \sigma_1 = \sigma_3 = 1 \) and \( \sigma_2 = 1.5 \). Simulations were conducted for \( d = 2, 3 \) and varying sample sizes.

The results in Figure \ref{fig:power_sample_d} illustrate that Wald and ANOVA tests show increasing power with larger samples, with ANOVA slightly outperforming Wald. The Percentile test is ineffective throughout. Higher dimension \( d = 3 \) offers minor gains in detection ability.
\end{example}
\begin{figure}[htb]
    \centering
    \includegraphics[width=0.8\textwidth, page=1]{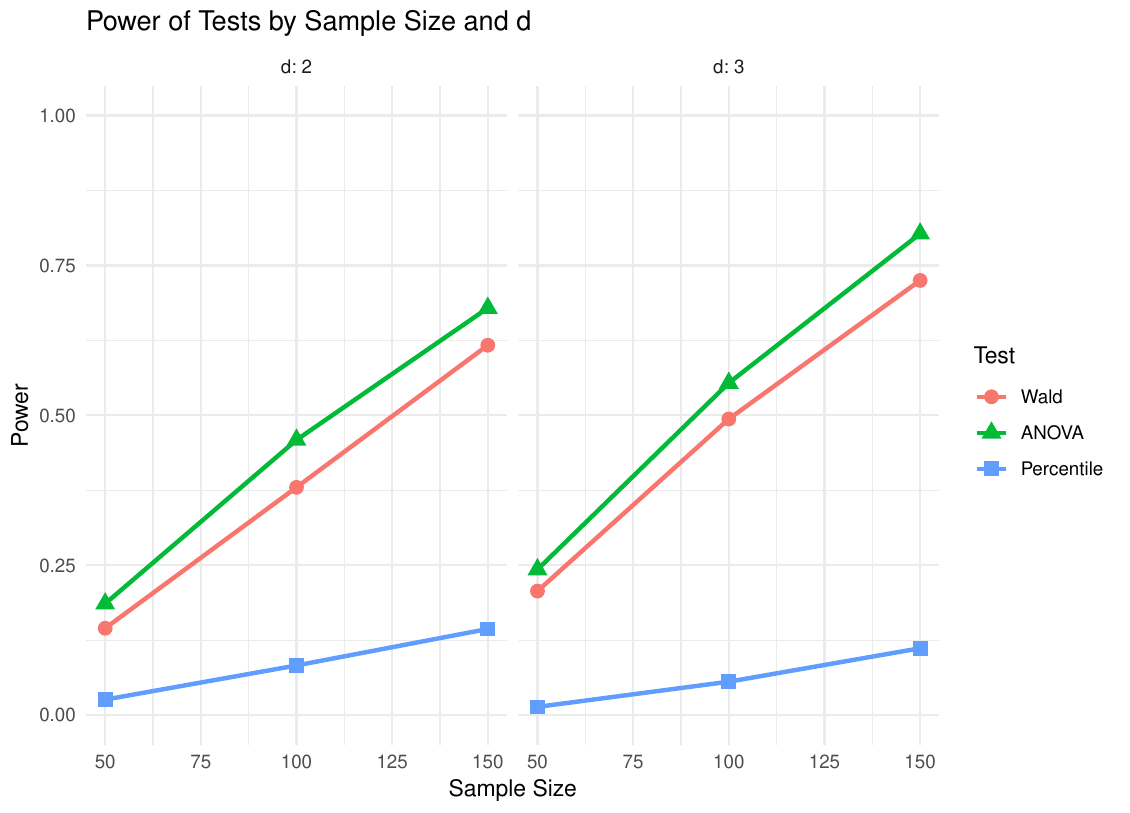}
  \caption{Empirical power for the tests in Section \ref{sec3.4} for \( k=3 \) groups with group-specific variances. Results shown for dimensions \( d = 2, 3 \). See Example \ref{ex8} for details.}
\label{fig:power_sample_d}
\end{figure}

\begin{example} \label{ex9}
{\bf (Power under Distributional Heterogeneity)} We consider three groups where two follow a multivariate normal distribution with mean vector \( \bm{\mu}_i = \mathbf{1}\in \mathbb{R}^d \) and  covariance matrix with unit variances and  0.25 correlations. In each group, the third component is replaced with a lognormal distribution with mean \(-0.35\) and variance \(0.7\), matching that of  original normal component.

Figure \ref{fig:power_sample_e} shows that Wald and ANOVA tests demonstrate strong power even for moderate sample sizes. The Percentile test again performs poorly. Results are consistent across \( d = 2 \) and \( d = 3 \).
\end{example}
\begin{figure}[htb]
    \centering
    \includegraphics[width=0.8\textwidth, page=1]{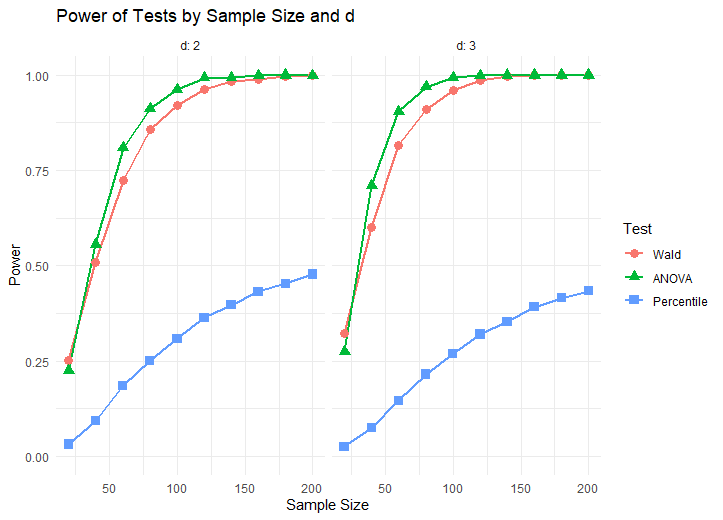}
    \caption{Empirical power for the tests in Section \ref{sec3.4} with a lognormal component replacing one variable. Results shown for \( d = 2, 3 \). See Example \ref{ex9} for details.}
\label{fig:power_sample_e}
\end{figure}
\begin{comment}
\begin{table}[h]
    \centering
    \renewcommand{\arraystretch}{1.2}
    \begin{tabular}{|c|c|c|c|c|}
        \hline
        \textbf{Sample Size} & \textbf{d} & \textbf{Wald Test} & \textbf{ANOVA Test} & \textbf{Percentile Test} \\
        \hline
        50  & 2 & 0.145 & 0.186 & 0.026000 \\
        50  & 3 & 0.207 & 0.243 & 0.013889 \\
        \hline
        100 & 2 & 0.380 & 0.459 & 0.083000 \\
        100 & 3 & 0.494 & 0.554 & 0.055889 \\
        \hline
        150 & 2 & 0.617 & 0.679 & 0.143833 \\
        150 & 3 & 0.725 & 0.803 & 0.111667 \\
        \hline
    \end{tabular}
    \caption{Power of the tests in Example \ref{ex8} (\( k=3 \), one group different).}
    \label{tab:rejection_rates3}
\end{table}
\end{comment}
\section{Case Study} \label{casestudy}
We return  to the motivating data  example and analyze the data on stable isotope measurements from muscle tissues of four fish species: Lake Whitefish (LKWF) ($n_1=69$), Broad Whitefish (BDWF) ($n_2=71$),  Arctic Cisco (ARCS) ($n_3=67$) and Least Cisco (LSCS) ($n_4=70$). We consider the three stable isotope ratios: Carbon isotope ratio ($\delta^{13}\mathrm{C}$), Nitrogen isotope ratio ($\delta^{15}\mathrm{N}$),  Sulfur isotope ratio ($\delta^{34}\mathrm{S}$). 

The data analysis yielded a p-value of $p < 0.0001$ for both the Wald-type test and the ANOVA-type test under the global null hypothesis. This strongly indicates that the niches of the four species are significantly different.

For a comprehensive analysis across all components, Figure \ref{fig:ci_fish} presents simultaneous $95 \%$ confidence intervals. Additionally, Table \ref{tab:variable_mapping} provides the corresponding variables and groups considered in each interval.
Note that niche overlaps close to $0.5$ indicate strong similarity in distributions, whereas values significantly deviating from 0.5 indicate substantial niche differentiation. In other words, lower values imply reduced probability that an observation from one species falls between two randomly drawn observations of another, suggesting distinct niches or lower inter-species competition. 

Our results indicate that the overlap measures for most of the variables are significantly different from the reference value ($1/2$). However, the $\delta^{13}$C isotopes do not show a significant difference from the reference value  for both Lake Whitefish and Broad Whitefish. Notably, only the elliptical-based confidence intervals have an upper bound slightly below $1/2$. This suggests a high level of competition for this isotope. In contrast, a low niche overlap is observed for the other isotopes, implying lower competition.

In their analysis, \citet{swanson_15} mentioned a high overlap between Lake Whitefish and Least Cisco but did not specify which variable contributed to this overlap. Additionally, their pairwise overlap analysis does not provide insight into how competition occurs within the overall group. They reported very low overlap for Broad Whitefish in all their pairwise comparisons. However, our detailed analysis reveals that this low overlap is present only for $\delta^{15}$N and $\delta^{34}$S, but not for $\delta^{13}$C. This significant discrepancy may arise because the assumption of multivariate normality is even more strongly violated for this group than for the others.
\begin{table}[h]
    \centering
    \begin{tabular}{lll}
        \toprule
        \textbf{Variable} & \textbf{Species} & \textbf{Isotope} \\
        \midrule
        Var 1  & ARCS & $\delta^{15}$N \\
        Var 2  & ARCS & $\delta^{13}$C \\
        Var 3  & ARCS & $\delta^{34}$S \\
        Var 4  & BDWF & $\delta^{15}$N \\
        Var 5  & BDWF & $\delta^{13}$C \\
        Var 6  & BDWF & $\delta^{34}$S \\
        Var 7  & LKWF & $\delta^{15}$N \\
        Var 8  & LKWF & $\delta^{13}$C \\
        Var 9  & LKWF & $\delta^{34}$S \\
        Var 10 & LSCS & $\delta^{15}$N \\
        Var 11 & LSCS & $\delta^{13}$C \\
        Var 12 & LSCS & $\delta^{34}$S \\
        \bottomrule
    \end{tabular}
    \caption{Linking of variable names with output in Figure \ref{fig:ci_fish}. The table defines the  variables combining the species that is compared with the overall group and the Isotope.}
    \label{tab:variable_mapping}
\end{table}

\begin{figure}[h]
    \centering
    \includegraphics[width=\textwidth]{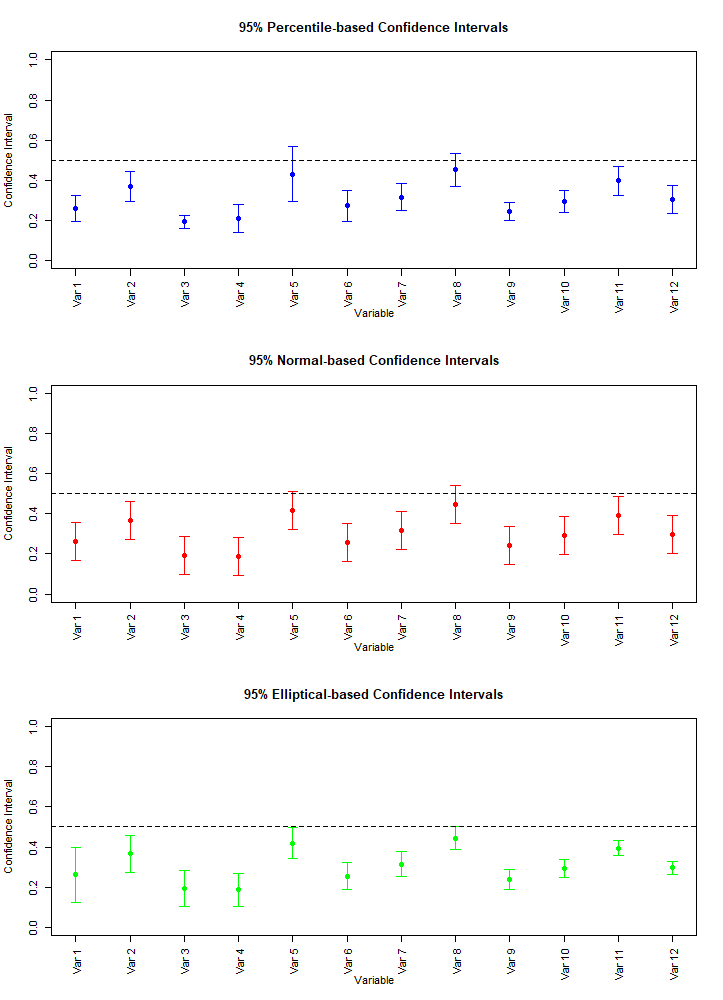}
    \caption{95\% Confidence Intervals: Percentile-based, Normal-based, and Elliptical-based.}
    \label{fig:ci_fish}
\end{figure}

\FloatBarrier
\section{Discussion} \label{discussion}
The methodological advancements presented in this paper significantly enhance the toolkit available for ecologists analyzing niche overlap in complex, real-world scenarios. Traditional niche overlap measures have been restricted primarily to univariate cases or relied on parametric assumptions that often fail to hold in practice, especially for ecological data which are often highly skewed and heavy-tailed. Our approach addresses these limitations by introducing a robust, nonparametric framework for quantifying and testing multivariate niche overlap. The proposed method demonstrates strong consistency and desirable asymptotic properties, which we confirmed through rigorous theoretical proofs and comprehensive simulation studies.

Our analysis of ecological data on stable isotope ratios from multiple fish species underscores the practical utility of this method. In particular, it revealed nuanced insights into species interactions and resource competition, highlighting both overlaps and distinct niches among the studied species, which prior parametric methods failed to capture accurately. This highlights the importance of utilizing nonparametric methods in ecological studies where multivariate normality assumptions may be inappropriate.

Future research should explore further refinements, including extensions to handle high-dimensional ecological data and temporally or spatially structured observations. Future work could also consider genuinely joint notions of niche overlap, for example based
on multivariate depth or hypervolume constructions, which define different estimands and
entail additional challenges for statistical inference and interpretability. Another interesting direction for future research is the longitudinal comparison of niche overlaps, enabling us to analyze multiple time points and gain deeper insights into ecological dynamics.  Overall, our method offers ecologists and researchers in related fields a powerful statistical tool to more accurately characterize ecological niches and interactions in multivariate contexts.

\section*{Acknowledgments}
The authors sincerely thank the Associate Editor and the expert reviewers for their careful reading of the manuscript and for their insightful and constructive comments, which have significantly improved the manuscript compared with its original version. The authors are also grateful to the Editor for the efficient and professional handling of the manuscript throughout the review process.  Jonas Beck gratefully acknowledges the Austrian Marshall Plan Foundation for the scholarship support that enabled his visit to the University of Kentucky in Spring 2024. He also sincerely thanks the Dr. Bing Zhang Department of Statistics for providing a welcoming and stimulating research environment, and is especially grateful to Dr. William Rayens, Chair of the Department, and Dr. Solomon Harrar for sponsoring his visit as a foreign scholar.  

\bibliographystyle{apalike}

\bibliography{niche_extend}

\appendix
\section{Proofs} \label{proofs}

\subsection{Derivation of the estimator} \label{A3}
Note that for even $n_i$ a split at the $l_i$-th observation corresponds to a split at the empirical median of the $i^{(s)}$-th sample.
Using ranks
we rewrite  the plug-in estimator: 
\begin{align*}
    & \int^{(\hat{F}^{(s)}_{i,n_i})^{-1}(1/2)}_{- \infty} \hat{H}^{(s)}_{N} d \hat{F}^{(s)}_{i,n_i}  
    = \frac{1}{n_i} \sum_{r=1}^{l_i} \hat{H}^{(s)}_N (X_{i,r}^{(s)})= \frac{1}{n_i} \sum_{r=1}^{l_i}  \sum_{j=1}^k \hat{F}^{(s)}_{j,n_j} (X_{i,r}^{(s)}) \\
    &=\frac{1}{n_i} \sum_{r=1}^{l_i}  \sum_{j=1}^{k}  \frac{1}{n_j} \sum_{r'=1}^{n_j} \mathbf{1}_{(-\infty,X_{i,r}^{(s)}]}(X_{j,r'}^{(s)}) 
    = \frac{1}{n_i N} \Bigl (  \sum_{r=1}^{l_i} R_{i,r}^{(s)} - \sum_{r=1}^{l_i} R_{i,r}^{(i)(s)} \Bigl ),
\end{align*}
if $n_i$ is even;
\begin{align*}
  & \int^{(\hat{F}_{i,n_i}^{(s)})^{-1}(1/2)}_{- \infty} \hat{H}^{(s)}_{N} d \hat{F}^{(s)}_{i,n_i} =
\frac{1}{n_i}
\left(
\sum_{r=1}^{l_i-1} 
\hat H^{(s)}_N\!\left(X^{(s)}_{i,r}\right)
+
\frac{1}{2}\,
\hat H^{(s)}_N\!\left(X^{(s)}_{i,l_i}\right)
\right)  \\
  & = \frac{1}{n_i N} \Bigl ( \sum_{r=1}^{l_i-1} R_{i,r}^{(s)} - \sum_{r=1}^{l_i-1} R_{i,r}^{(i)(s)} +\frac{1}{2} R_{i,l_i}^{(s)} - \frac{1}{2} R_{i, l_i}^{(i)(s)} \Bigl),
\end{align*}
if $n_i$ is odd. Similarly
\begin{align*}
    \int_{(\hat{F}_{i,n_i}^{(s)})^{-1}(1/2)}^{\infty} \hat{H}^{(s)}_{N} d\hat{F}^{(s)}_{i,n_i} = \begin{cases}
       \frac{1}{n_i N} \Bigl (  \sum_{r=l_i+1}^{n_i} R_{i,r}^{(s)} - \sum_{r = l_i +1}^{n_i} R_{i,r}^{(i)(s)} \Bigl), \text{if } n_i \text{ is even} \\
        \frac{1}{n_i N} \Bigl ( \sum_{r=l_i+1}^{n_i} R_{i,r}^{(s)} - \sum_{r = l_i +1}^{n_i} R_{i,r}^{(i)(s)} +\frac{1}{2} R_{i,l_i}^{(s)} - \frac{1}{2} R_{i,l_i}^{(i)(s)} \Bigl ),  \text{if } n_i \text{ is odd}
    \end{cases} 
\end{align*}

Therefore,
\begin{align*}
   & I(\hat{H}^{(s)}_{N},\hat{F}^{(s)}_{i,n_i}) 
    = \frac{2}{n_i N} \Bigl (\sum_{r=l_i+1}^{n_i} R_{i,r}^{(s)} - \sum_{r = l_i +1}^{n_i} R_{i,r}^{(i)(s)} - \sum_{r=1}^{c_i} R_{i,r}^{(s)} + \sum_{r=1}^{c_i} R_{i,r}^{(i)(s)} \Bigl ),
\end{align*}
where $c_i=l_i$ for $n_i$ even and $c_i=l_i-1$ for $n_i$ odd.
\subsection{Proof of Theorem \ref{thm:asym:2}} \label{A4}
    
It can be shown that \cite{langthaler}
\[ 
\mathcal{I}(\mathbf{H},\mathbf{F}_i) = \lambda_i \mathcal{I}(\mathbf{F}_i,\mathbf{F}_i) + (1- \lambda_i) \mathcal{I}(\mathbf{H}_{-i},\mathbf{F}_i) = \frac{\lambda_i}{2} \mathbf{1}_d + (1- \lambda_i) \mathcal{I}(\mathbf{H}_{-i},\mathbf{F}_i),
\]
where \begin{align*}
    \mathbf{H} = \lambda_i \mathbf{F}_i + (1- \lambda_i) \mathbf{H}_{-i}  \text{ and } \mathbf{H}_{-i} := \frac{1}{1-\lambda_i}\sum_{j \neq i}^k \lambda_j \mathbf{F}_j 
\end{align*}
Here, we utilized the linearity of $I(.,.)$ in the first argument.

Now consider:
\begin{align}
 \sqrt{N} \left( \widehat{\mathcal{I}}_{N,n_i}(\mathbf{H},\mathbf{F}_i) - \mathcal{I}(\mathbf{H},\mathbf{F}_i) \right) &= 
\sqrt{N } \frac{N-n_i}{N} \left( \widehat{\mathcal{I}}_{N-n_i,n_i}(\mathbf{H}_{-i},\mathbf{F}_i) - \mathcal{I}(\mathbf{H}_{-i},\mathbf{F}_i) \right)  \notag \\
&+ \sqrt{N } \frac{n_i}{N} \left(\widehat{\mathcal{I}}_{n_i,n_i}(\mathbf{F_{i}},\mathbf{F}_i) - \mathcal{I}(\mathbf{F_{i}},\mathbf{F}_i) \right) \label{star}
\end{align}
Since $\sqrt{N } \frac{n_i}{N} \left(\widehat{\mathcal{I}}_{n_i,n_i}(\mathbf{F_{i}},\mathbf{F}_i) - \mathcal{I}(\mathbf{F_{i}},\mathbf{F}_i) \right) = o_{p}(1)$, we only need to consider the asymptotic distribution of the first term in \eqref{star}. Hence, the asymptotic distribution is driven by terms involving samples independent of $F_i$.

The map
\[
\psi(F_i, F_j) = F_i \circ F_j^{-1}
\]
for distribution functions \( F_i \) and \( F_j \) defined on \( \mathbb{R} \) is Hadamard-differentiable tangentially to \( D[a, b] \times C[a, b] \) by \citep[][Comment 4 in Section 3.10]{vaart2023statistical}. Further, the map $\phi: \mathbb{D} \rightarrow \mathbb{E}$ defined by:
\begin{align*} 
 \phi(f)=  \int_0^1 f \biggl(1- \frac{\alpha}{2} \biggl )d \alpha- \int_0^1 f \biggl(\frac{\alpha}{2} \biggl) d \alpha,
\end{align*}
is Hadamard-differentiable. Since the composition of two Hadamard-differentiable maps is again Hadamard-differentiable, by chain rule \citep[][Theorem 20.9]{vaart_1998}, the map $\phi \circ \psi$ is Hadamard-differentiable. Noting that $$I(F_i^{(s)},F_j^{(s)})=\phi \circ \psi(F_i^{(s)},F_j^{(s)}),$$ we have Hadamard-differentiability of the niche overlap function. The extension to the multivariate case is straight forward. 
By the Donsker Theorem \citep[][Theorem 19.3]{vaart_1998} we know that
$
    \sqrt{N}(\hat{F}^{(s)}_{i,n_i}-F^{(s)}_{i,n_i}) \  
$
converges in distribution to a Gaussian process for all $i=1, \ldots, k,s= 1, \ldots, d$. Applying the functional delta method \citep[][Theorem 3.9.4]{van1996weak} proves the asymptotic of \eqref{eq:conv,mult}.

The conditional central limit theorem holds in outer probability for each bootstrapped empirical distribution function,
$
    \sqrt{N}(\hat{F}^{(s)*}_{i,n_i}-\hat{F}^{(s)}_{i,n_i}) 
$
 for all $i=1, \ldots, k,s= 1, \ldots, d$ by \citep[][Theorem 3.6.1]{van1996weak}. Moreover, because the bootstrap resamples $d$-dimensional vectors, the collection of
marginal bootstrap empirical processes $(\widehat F^{(1)*}_{i,n_i},\ldots,\widehat F^{(d)*}_{i,n_i})$
is generally dependent across components when the components of $X_{i,r}$ are dependent,
and this dependence is therefore reproduced in the bootstrap covariance structure (hence in $\Sigma$),
so that the conditional weak convergence holds jointly in $\mathbb{R}^{kd}$.

By applying the Delta-method for empirical bootstrap processes \citep[][Theorem 3.9.11]{van1996weak}, it follows (analogous to the unconditional case) that \eqref{eq: boot,mult} converges conditionally on the data in outer probability to the same limiting distribution as \eqref{eq:conv,mult}.

\clearpage
\section*{Supplementary Material}
\addcontentsline{toc}{section}{Supplementary Material}

\setcounter{section}{0}
\renewcommand{\thesection}{S\arabic{section}}

\section{Two-Sample Setting} \label{sec:2}
In this section, we propose nonparametric niche overlap  \citep{ParkinsonKutilKupplerJunkerTrutschnigBathke+2018, langthaler} in the multivariate case for two-samples.
\subsection{Statistical Model}
Let the independent and identically distributed random vectors 
\[
\mathbf{X_{k}} = (X_{k}^{(1)}, \ldots, X_{k}^{(d)})^\top \sim \mathbf{F},
\]
for $k=1,\ldots,n$ represent the data for the first sample, 
 where \(X_{k}^{(s)}\) denotes the observation on the \(s\)-th endpoint of the \(k\)-th subject in the first sample. Similarly, let \[
\mathbf{Y_{k}} = (Y_{k}^{(1)}, \ldots, Y_{k}^{(d)})^\top \sim \mathbf{G},
\]for $k=1,\ldots,m$ be identically and independently distributed random vectors representing the data from the second sample, where \(Y_{k}^{(s)}\) denotes the observation on the  $s$th endpoint of the \(k\)-th subject. We assume that the two samples are mutually independent.   Let \(N = n+ m \) be the total sample size for each endpoint. The marginal distributions  are denoted by $F^{(1)}, \ldots, F^{(d)}$ and $G^{(1)}, \ldots, G^{(d)}$, i.e., \(X_{k}^{(s)} \sim F^{(s)}\) for \( s \in \{1, \ldots, d\}\), where \(F^{(s)}(x) = \text{Pr}(X_{k}^{(s)} < x) \) and \(Y_{k}^{(s)} \sim G^{(s)}\) for  \(s \in \{1, \ldots, d\}\), where \(G^{(s)}(x) = \text{Pr}(Y_{k}^{(s)} < x) \). The marginal distributions \(F^{(s)}\) and \(G^{(s)}\) are assumed to be absolutely continuous. 

To define the overlap index, we first have to introduce the random variables $Y^{(s)}_{\text{low}} \sim G^{(s)}_{\text{low}}$ independently of $Y^{(s)}_{\text{up}} \sim G^{(s)}_{\text{up}}$, where $G^{(s)}_{\text{low}}$ and $G^{(s)}_{\text{up}}$ are the distribution of $Y_i^{(s)}$ conditional on the events $Y \leq( G^{(s)})^{-1} (0.5)$ and $Y > ( G^{(s)})^{-1} (0.5)$, respectively.  That is, 
\begin{align*}
G^{(s)}_{\text{low}}(x) &= 
\begin{cases}
2 G^{(s)}(x), \ &x < ( G^{(s)})^{-1} (0.5) \\
1,      &x \geq ( G^{(s)})^{-1} (0.5)
\end{cases}
\end{align*}
and
\begin{align*}
G^{(s)}_{\text{up}}(x) &= 
\begin{cases}
0, \ &x < ( G^{(s)})^{-1} (0.5) \\
2G^{(s)}(x)-1,      &x \geq ( G^{(s)})^{-1} (0.5)
\end{cases}.
\end{align*}
$G^{(s)}_{\text{low}}$ is the distribution function $G^{(s)}$ conditioned on $Y^{(s)} \leq (G^{(s)})^{-1}(0.5)$, and 
$G^{(s)}_{\text{up}}$ is $G^{(s)}$ conditioned on 
$Y^{(s)} > (G^{(s)})^{-1}(0.5)$.
The random variables $X_{\text{low}}^{(s)} \sim F_{\text{low}}^{(s)}$ and $X_{\text{up}}^{(s)} \sim F_{\text{up}}^{(s)}$ are analogously defined, by switching the roles of $G$s and $F$s.

The multivariate overlap index is defined by the vector
\begin{align*}
    \mathcal{I}(\mathbf{F},\mathbf{G})= (I(F^{(1)},G^{(1)}), \ldots, I(F^{(d)},G^{(d)}))^\top,
\end{align*}
where 
\begin{equation} \label{eq:mvOverlapIndex}
    I(F^{(s)},G^{(s)}) =  P(Y_{\text{low}}^{(s)}< X^{(s)} < Y_{\text{up}}^{(s)}).
\end{equation}  
Thus, the index is large when most of $F^{(s)}$ lies in the central region of $G^{(s)}$,
and small when $F^{(s)}$ is concentrated far below or above $G^{(s)}$’s median. Equivalently, for each component it represents the probability that an observation from the first distribution falls between two observations from the second—one drawn from below its median and one from above.

The overlap index can be rewritten  \citep{ParkinsonKutilKupplerJunkerTrutschnigBathke+2018, langthaler}  as
\begin{align*}
I(F^{(s)},G^{(s)}) &= \int_0^1 F^{(s)} \circ {G^{(s)}}^{-1}(1- \alpha/2)d \alpha - \int_0^1 F^{(s)} \circ {G^{(s)}}^{-1}(\alpha/2) d \alpha   \\
&= 2 \biggl (\int_{{G^{(s)}}^{-1}(1/2)}^{\infty} F^{(s)} dG^{(s)} - \int^{{G^{(s)}}^{-1}(1/2)}_{- \infty} F^{(s)} d G^{(s)}  \biggl ). 
\end{align*} 
The first term equation shows the close relation to the ROC-curve.

Note that $F^{(s)}=G^{(s)}$ implies the corresponding overlap index $I(F^{(s)},G^{(s)})=1/2$. However, the converse is not true. In addition, the overlap index is not symmetric. Because of this asymmetry, $I(F^{(s)},G^{(s)})$ and $I(G^{(s)},F^{(s)})$ can capture different facets of the relationship between $F^{(s)}$ and $G^{(s)}$, so the order of the arguments should be chosen deliberately (or both directions reported). Further, it holds that $$0\le I(F^{(s)},G^{(s)})+ I(G^{(s)},F^{(s)}) \le 1$$ and the upper bound is attained when the medians of the two distributions are equal \citep[][Lemma 2.10]{ParkinsonKutilKupplerJunkerTrutschnigBathke+2018}.

\subsection{Estimation} \label{sec:est2}

We estimate the functional $\mathcal{I}(\mathbf{F},\mathbf{G})$ by  plugging in the estimators of $F^{(s)}$ and $G^{(s)}$ in \eqref{eq:mvOverlapIndex} as
\begin{equation*}%\label{eq:mvnIndexExtimator}
    \widehat{\mathcal{I}}_{n,m}(\mathbf{F},\mathbf{G})= (I(\hat{F}_{n}^{(1)},\hat{G}_{m}^{(1)}), \ldots, I(\hat{F}_{n}^{(d)},\hat{G}_{m}^{(d)}))^\top,
\end{equation*}
where
\begin{align*} 
& I(\hat{F}^{(s)}_{n},\hat{G}^{(s)}_{m}) \notag \\ 
& = \int_0^1 \hat{F}^{(s)}_{n} \circ (\hat{G}^{(s)}_{m})^{-1}(1- \alpha/2)d \alpha - \int_0^1 \hat{F}^{(s)}_{n} \circ (\hat{G}^{(s)}_{m})^{-1}(\alpha/2) d \alpha \notag \\
& = 2 \biggl (\int_{(\hat{G}^{(s)}_{m})^{-1}(1/2)}^{\infty} \hat{F}^{(s)}_{n} d\hat{G}^{(s)}_{m} - \int^{(\hat{G}^{(s)}_{m})^{-1}(1/2)}_{- \infty} \hat{F}^{(s)}_{n}  d \hat{G}^{(s)}_{m} \biggl )
\end{align*}
$\hat F_n^{(s)}$ and $\hat G_m^{(n)}$ are marginal empirical distributions of the $s$-th endpoint in two samples, respectively, $(\hat{G}^{(s)}_{m})^{-1}$ is the empirical quantile function of $G_n^{(s)}$ defined by $$(\hat{G}^{(s)}_{m})^{-1}(t)= \inf  \{ y: \hat{G}^{(s)}_{m}(y) \geq t \}$$for   $s=1, \ldots, d$. 
This estimator is strongly consistent, which follows directly by the strong consistency of the component-wise niche overlap estimators proved in \cite{ParkinsonKutilKupplerJunkerTrutschnigBathke+2018} and the fact that $d$ is finite. 

To formulate a computationally convenient form of our estimator $I(\hat{F}^{(s)}_{n},\hat{G}^{(s)}_{m})$, we assume without loss of generality that $Y_1^{(s)} \leq Y_2^{(s)} \leq \ldots \leq Y^{(s)}_{m}$. We denote by $L$ the largest integer smaller than or equal to $(m+1)/2$.  The $L$ observations below the median are $Y^{(s)}_1, \ldots, Y^{(s)}_L$. We denote their ranks in the combined sample by  $R_1^{Y^{(s)}<}, \ldots, R_L^{Y^{(s)}<}$.  Similarly, we denote the  ranks of the remaining observations  by $R_{L+1}^{Y^{(s)}>}, \ldots, R_{m}^{Y^{(s)}>}$. The corresponding rank sums are defined by 
\begin{align*}
R_{\cdot}^{Y^{(s)}<} = \sum_{i=1}^{L} R_i^{Y^{(s)}<}, \quad \text{and} \quad R_{\cdot}^{Y^{(s)}>} = \sum_{i=L+1}^{m} R_i^{Y^{(s)}>}.
\end{align*}

The    estimator $ I(\hat{F}^{(s)}_{n},\hat{G}^{(s)}_{m})$ of the niche overlap for the $s$-th endpoint can then, as shown in Lemma 2.11 in \cite{ParkinsonKutilKupplerJunkerTrutschnigBathke+2018}, be expressed as 
\begin{align*}
I(\hat{F}^{(s)}_{n},\hat{G}^{(s)}_{m})=
\frac{2}{m n} (R_{\cdot}^{Y^{(s)}>}-R_{\cdot}^{^{(s)}Y<} ) + \frac{1}{2} c,
\end{align*}
where $c=- m/n$ for $m$ even and $c \approx -m/n$ for $m$ odd.

\subsection{Asymptotic Distribution and Resampling}
In this subsection we use the similarity of the overlap index to the ROC curve, denoted  as $F^{(s)} \circ {G^{(s)}}^{-1}(u)$, to derive an asymptotic distribution for the estimator of the multivariate niche overlap index. 
Therefore, we need some technical conditions: 
\begin{assumption} \label{assum_1}
The ratio of the sample sizes $\frac{m}{n}\to \nu>0$ as the total sample size $N=n+m\to \infty$.  \end{assumption}

One can see that the slope of the ROC curve is given by
\begin{align*}
F^{(s)} \circ {G^{(s)}}^{-1}(u)=\frac{f^{(s)}({G^{(s)}}^{-1}(u))}{g^{(s)}({G^{(s)}}^{-1}(u))}.
\end{align*}
\begin{assumption}\label{assum_2}
The slope of the ROC, $F^{(s)} \circ {G^{(s)}}^{-1}(u)$, is bounded on  the interval $(a, b)$  for any $0<a < b<1$ and any $s=1,\ldots,d$.
\end{assumption}
Assumption \ref{assum_1} is a standard proportional divergence requirement for sample sizes. It guarantees that the sample sizes are of the same order of magnitude and thus their contributions to the dataset are comparable.  The bounded-slope condition of Assumption \ref{assum_2} is equivalent to the assumption that the slope of the curve $F^{(s)} \circ {G^{(s)}}^{-1}$ is finite for all $u \in (0,1)$, i.e.,
\begin{align*}
\ \ \Biggl | \frac{f^{(s)}({G^{(s)}}^{-1}(u))}{g^{(s)}({G^{(s)}}^{-1}(u))} \Biggl | < \infty,\quad \forall u \in (0,1).
\end{align*}

Under Assumptions \ref{assum_1} and \ref{assum_2}, it can be shown \cite{hsieh_turnbull} that the plug-in estimator of the ROC curve is strongly consistent.

We now state the asymptotic distribution of the overlap index estimator.
\begin{theorem}
\label{th:asymptotic}
    Under Assumptions \ref{assum_1}  and \ref{assum_2}, \begin{equation}
    \sqrt n(\widehat{\mathcal{I}}_{n,m}(\mathbf{F},\mathbf{G})-\mathcal{I}(\mathbf{F},\mathbf{G}))\overset{d}{\to}N_d( \bm 0,\Sigma), \label{eq:conv}
     \end{equation}

where $N_d(\bm \mu,\Sigma)$ stands for the multivariate normal distribution with mean vector  $\bm \mu$ and a positive semidefinite covariance matrix $\Sigma$.
\end{theorem}
\begin{proof}
See Section \ref{A1}
\end{proof}

In principle, an explicit formula for the asymptotic covariance matrix~$\Sigma$ could be derived, as the proof of Theorem~\ref{th:asymptotic} shows that its entries can be expressed in terms of the covariance between two Brownian bridges, for which closed-form expressions are available. However, the expression  would be very involved and, thus, would not be particularly insightful for practical application.  Therefore, we will employ  a bootstrap  strategy. In what follows, we devise a bootstrap approximation for the distribution of the multivariate overlap estimator. 

Let $X_1^*,\ldots,X_n^*$ be an iid sample from the empirical distribution $\hat{\mathbf F}_n$
on $\mathbb{R}^d$, obtained by sampling with replacement from the observed vectors
$X_1,\ldots,X_n\in\mathbb{R}^d$. Analogously, let $Y_1^*,\ldots,Y_m^*$ be an iid sample
from the empirical distribution $\hat{\mathbf G}_m$ on $\mathbb{R}^d$, obtained by sampling
with replacement from $Y_1,\ldots,Y_m\in\mathbb{R}^d$. Writing
$X_j^*=(X_j^{(1)*},\ldots,X_j^{(d)*})^\top$ and $Y_j^*=(Y_j^{(1)*},\ldots,Y_j^{(d)*})^\top$,
we define, for each component $s=1,\ldots,d$, the bootstrap marginal empirical distribution
functions by
\[
\hat{F}^{(s)*}_{n}(t) := \frac{1}{n}\sum_{j=1}^{n}\mathbf{1}_{(-\infty,t]}\!\left(X^{(s)*}_{j}\right),
\qquad
\hat{G}^{(s)*}_{m}(t) := \frac{1}{m}\sum_{j=1}^{m}\mathbf{1}_{(-\infty,t]}\!\left(Y^{(s)*}_{j}\right).
\] 
We define the bootstrap version of $\widehat{\mathcal{I}}_{n,m}(\mathbf{F},\mathbf{G})$ by 
 $$ \widehat{\mathcal{I}}_{n,m}^*(\mathbf{F},\mathbf{G})= (I(\hat{F}_{n}^{(1)*},\hat{G}_{m}^{(1)*}), \ldots, I(\hat{F}_{n}^{(d)*},\hat{G}_{m}^{(d)*}))^\top.$$
To show that this bootstrap strategy works, we prove that the bootstrap process converges to the same limit distribution as our estimator.

\begin{theorem} \label{thmboot}
    Under Assumptions \ref{assum_1} and \ref{assum_2}, it holds that
  $  
    \sqrt n \{ \widehat{\mathcal{I}}^*_{n,m}(\mathbf{F},\mathbf{G})- \widehat{\mathcal{I}}_{n,m}(\mathbf{F},\mathbf{G}) \}
 $  
    converges in outer probability conditionally on the data to the same limit distribution as \eqref{eq:conv}.
\end{theorem}
\begin{proof}
See Section \ref{A2}
\end{proof}

\subsection{Hypothesis Test} \label{sec2.4}
As mentioned before, in case of equal distributions in one component the overlap index takes the marginal value $1/2$. Therefore, it is natural to test the null hypothesis 
\begin{align}\label{eq:mHypothesis}
    H_0 : \mathcal{I}(\mathbf{F},\mathbf{G})=\frac{1}{2} \mathbf{1}_d \text{ versus } H_1: \mathcal{I}(\mathbf{F},\mathbf{G}) \neq \frac{1}{2} \mathbf{1}_d,
\end{align}
where $\mathbf{1}_d=(1, \ldots, 1)^T$ is a vector of all $1$'s

By Theorem \ref{th:asymptotic}, under $H_0: \mathcal{I}(F,G) = \frac{1}{2} \bm 1_d$,  and the Continuous Mapping Theorem and under the additional assumption of a nonsingular covariance matrix, the distribution of the quadratic form
   \[
Q_n = n \cdot\ (\widehat{\mathcal{I}}_{n,m}(\mathbf{F},\mathbf{G})- \frac{1}{2} \bm{1_d})^T \Sigma^{-1} (\widehat{\mathcal{I}}_{n,m}(\mathbf{F},\mathbf{G})-\frac{1}{2} \bm{1_d})
\]
  converges to a central $\chi^2$-distribution with $d$ degrees of freedom. As we have shown in Theorem \ref{thmboot}, the bootstrap process converges, so we can put in the bootstrapped covariance matrix $\Sigma^*$ and the resulting test statistic 
    \[
\hat{Q}_n = n \cdot\ (\widehat{\mathcal{I}}_{n,m}(\mathbf{F},\mathbf{G})- \frac{1}{2} \bm{1_d})^T \Sigma^{*^{-1}} (\widehat{\mathcal{I}}_{n,m}(\mathbf{F},\mathbf{G})-\frac{1}{2} \bm{1_d})
\]converges to the same limiting distribution.
Therefore, the test statistic $\hat{Q}_n$ is again central $\chi^2$-distributed with $d$ degrees of freedom. 
If $\Sigma$ is singular, one may either reduce the contrast to a full-rank sub-contrast
or replace $\Sigma^{-1}$ by the Moore--Penrose generalized inverse $\Sigma^{+}$,
leading to a Wald-type statistic based on the effective rank.

It is well known in the literature that this approximation is too liberal for small sample size \citep{cui2021nonparametric, brunner2002multivariate}. Our simulations also confirm this assertion. The use of the inverse of the estimated covariance as the middle matrix of the quadratic form is mainly responsible for this fragile behavior of the Wald-type test. To overcome this problem,  we   adopt the ANOVA-type statistic first proposed for parametric setting \cite{box} and later extended for univariate and multivariate nonparametric settings \cite{brunner97, brunner2002multivariate}. The idea of this statistic is to use the trace of the estimated covariance matrix in the denominator, thereby making  the statistic resemble the classical ANOVA-statistic as the ratio of two quadratic forms. Therefore, 
 under the null hypothesis,  the ANOVA-type statistic
\begin{align*}
F_n &= \frac{n}{\operatorname{tr}({\hat\Sigma^*})} (\widehat{\mathcal{I}}_{n,m}(\mathbf{F},\mathbf{G}) - \frac{1}{2} \mathbf{1}_d)^{\top} (\widehat{\mathcal{I}}_{n,m}(\mathbf{F},\mathbf{G}) - \frac{1}{2} \mathbf{1}_d) \\
&= \frac{n}{\operatorname{tr}(\hat{\Sigma^*})} \sum_{l=1}^{d} \left( \mathcal{I}(\hat{F}^{(l)},\hat{G}^{(l)}) - \frac{1}{2} \right)^2
\end{align*}
has approximate central \( F(\hat{\nu}, \infty) \) distribution, where \( \hat{\nu} = [\operatorname{tr}({\hat \Sigma^*})]^2/\operatorname{tr}(\hat \Sigma^*{}^2)\). For more details on this F-approximation we refer the reader to  \citet{brunner2002multivariate}. Our simulations show that the ANOVA-type statistic is better for small sample sizes in keeping the level of the test, while still having similar power compared to the Wald-type statistic.

Testing the hypothesis \eqref{eq:mHypothesis} is equivalent to testing simultaneously the null hypotheses that the niche overlap equals $1/2$ in each component.  In view of this, we consider the multivariate test statistic $T= (T_1,\ldots,T_d),$ where 
 \begin{align*} 
    T_s =  \sqrt{n} \frac{(I(\hat{F}^{(s)}_{n}, \hat{G}^{(s)}_{m}) - \frac{1}{2})}{\sigma^{(s)*}}
    \end{align*}
   for $s=1,\ldots,d$ and   $\sigma^{(s)*}$ is the bootstrapped standard deviation of the $s$-th component.  By Theorem \ref{th:asymptotic}, under the null hypothesis, $T$ has asymptotic  multivariate normal distribution with mean vector $\mathbf{0}$ and covariance matrix $P^*$, where $P^*$ is the correlation matrix of the empirical bootstrap sample.   
According to the so-called Max-T test \citep{multiple-contrast, KonietschkeBösigerBrunnerHothorn+2013+63+73,GUNAWARDANA2019165}  the null hypothesis \eqref{eq:mHypothesis} may be rejected at level $\alpha$, if 
\begin{align*} 
    T_0= \max (|T_1|, \ldots, |T_d|) \geq z(1-\alpha, P^*),
\end{align*}
where $z(1-\alpha, P^*)$ the equicoordinate quantiles. The quantiles can computed with the \textsc{R}-package \textsc{mvtnorm} \citep{mvtnorm, genz_bretz}.   This test, however, tends to lack power as shown later in the Simulation section. Notwithstanding this limitation, an advantage of the Max-T test is that it is possible to detect the components where the difference lies  by comparing  each component of $T$ with $z(1-\alpha, P^*)$. The vector $T$  may be analyzed for various  contrast matrices which could make an interesting new direction of research, similar to the works on multiple contrast test for the relative effect \citep{GUNAWARDANA2019165,dobler}.

\subsection{Confidence Intervals} \label{sec_ci}
Similar to hypothesis testing there are a multitude of approaches to derive confidence regions and simultaneous confidence intervals. Given that the estimator of the niche overlap has an asymptotic multivariate normal distribution  (Theorem \ref{th:asymptotic}) and  a bootstrap approximation is valid (Theorem \ref{thmboot}), the test based on Wald-type statistic can be inverted to derive an asymptotic elliptical $1-\alpha$ confidence region, i.e., the set of all $d$-dimensional vectors $ {\bf v}=(v_1, \ldots, v_d)^\top$ satisfying
\begin{align*}
   \frac{1}{\sqrt{n}} \biggl ( \widehat{\mathcal{I}}_{n,m}(\mathbf{F},\mathbf{G})- {\bf v} \biggl )^T  \Sigma^*_{n,m} \biggl( \widehat{\mathcal{I}}_{n,m}(\mathbf{F},\mathbf{G})- {\bf v} \biggl ) \leq q_{\chi^2_2}(1-\alpha), 
    \end{align*}
    where $q_{\chi^2_d}(1-\alpha)$ is the $1-\alpha $ quantile of a $\chi^2$ distribution with d degrees of freedom and  $\Sigma^*_{n,m}$ is the sample covariance matrix of the bootstrap estimates.
   This construction is already well known in parametric statistics.

The statistic $T$ can be used to obtain asymptotic compatible \( (1 - \alpha) \)   (compatible in the sense that any parameter vector inside the joint confidence region satisfies all componentwise intervals at the same time) Simultaneous Confidence Intervals (SCIs) for the niche overlap as 
\[
 \widehat{\mathcal{I}}_{n,m}(\mathbf{F},\mathbf{G}) \pm z(1 - \alpha, \Sigma^* )
 ,\]
Here also we calculate the equi-coordinate multivariate normal quantiles by the \textsc{R}-package \textsc{mvtnorm} \citep{mvtnorm, genz_bretz} and the methods described therein.
 
 In small samples, the multivariate normal distribution may provide a poor approximation for the distribution of $T$.  Using Theorem \ref{thm:asym:2}, one may use empirical bootstrap quantiles with a Bonferroni correction to construct valid simultaneous confidence intervals as 
 
\begin{align*}
[\widehat{\mathcal{I}}_{n,m}(\mathbf{F},\mathbf{G}) - \frac{1}{\sqrt{n}} \mathbf{z}^*_{\alpha/(2d)}, \widehat{\mathcal{I}}_{n,m}(\mathbf{F},\mathbf{G})+ \frac{1}{\sqrt{n}} \mathbf{z}^*_{1-\alpha/(2d)}  ]
\end{align*}
where $\mathbf{z}^*_{1-\alpha/(2d)}=(z^*_{I(\hat{F}^{(1)}_{n}, \hat{G}^{(1)}_{m}), 1-\alpha/(2d)},\ldots,  z^*_{I(\hat{F}^{(d)}_{n}, \hat{G}^{(d)}_{m}), 1-\alpha/(2d)})^\top$ 
denote the quantiles of the empirical bootstrap distribution of our estimator $\widehat{\mathcal{I}}_{n,m}(\mathbf{F},\mathbf{G})$. 

 All three methods will be compared in our simulation study in Section \ref{sim}.
 \section{Additional Numerical Examples} \label{s2}
\subsection{Confidence Intervals}
The aim of the numerical studies in this section is to evaluate the performance of the confidence region and simultaneous confidence interval methods in the two-sample case. The interval length is calculated as the mean of all component-wise interval lengths. We compare the three methods proposed in Section \ref{sec_ci}; namely,  (i) the empirical bootstrap quantiles with a Bonferroni correction (Bonf),  (ii) the  equi-coordinate multivariate normal quantiles (MVT) and (iii) the asymptotic elliptical confidence region (Ellipse) as described in Section \ref{sec_ci}. 
\begin{example} \label{ex1}
This example is concerned with the comparison of two \( d \)-dimensional multivariate normal distributions with mean vectors \( \bm \mu^{(g)} = \mathbf{0} \in \mathbb{R}^d \) and positive definite covariance matrices $\Sigma^{(g)}$ for $g=1,2$. For each group \(g\in \{1,2\}\), the off diagonal entries of $\Sigma^{(g)}$ are set to $0.25$ and the diagonal entries are set to $\sigma^2{} ^{(1)}=2$ and  $\sigma^2{} ^{(2)}=1$ respectively.
 The true value of the overlap index is ${\rm I}(\bm F,\bm G)=0.6082\cdot \bm 1_d$.

All the methods generally maintain high coverage close to the nominal level, but the Ellipsoid method tends to yield slightly higher coverage, especially as the dimension increases. The Bonferroni and MVT methods produce shorter confidence intervals compared to the Ellipsoid method, indicating they could be more precise (Table \ref{table1}). 

The observed deviations from nominal coverage are attributable to slow finite-sample
convergence of Wald-type procedures when the covariance matrix is estimated, particularly
in moderate dimensions.The Bonferroni procedure is not uniformly conservative in finite samples here, since it is
applied to marginal percentile bootstrap intervals, which may exhibit slight undercoverage
for moderate sample sizes and skewed distributions.
\end{example}
\begin{table}[htb]
\centering
\begin{tabular}{|c|c|c|c|c|c|c|c|}
\hline
\multirow{2}{*}{n=m} & \multirow{2}{*}{Dim.} & \multicolumn{3}{c|}{Coverage Probability} & \multicolumn{3}{c|}{Interval Length} \\ \cline{3-8} 
                     &                            & Bonf & MVT & Ellipse & Bonf & MVT & Ellipse \\ \hline
\multirow{4}{*}{50}  & 2                          & 0.964      & 0.960   & 0.976         & 0.271      & 0.275   & 0.299         \\ \cline{2-8} 
                     & 3                          & 0.980      & 0.972   & 0.982         & 0.289      & 0.286   & 0.341         \\ \cline{2-8} 
                     & 4                          & 0.980      & 0.970   & 0.994         & 0.301      & 0.307   & 0.376         \\ \cline{2-8} 
                     & 5                          & 0.979      & 0.975   & 0.996         & 0.309      & 0.317   & 0.406         \\ \hline
\multirow{4}{*}{100} & 2                          & 0.969      & 0.961   & 0.975         & 0.185      & 0.187   & 0.204         \\ \cline{2-8} 
                     & 3                          & 0.965      & 0.959   & 0.989         & 0.198      & 0.200   & 0.233         \\ \cline{2-8} 
                     & 4                          & 0.971      & 0.970   & 0.995        & 0.206      & 0.209   & 0.257         \\ \cline{2-8} 
                     & 5                          & 0.972      & 0.969   & 0.998         & 0.212      & 0.215   & 0.277         \\ \hline
\multirow{4}{*}{200} & 2                          & 0.948      & 0.948   & 0.971         & 0.128      & 0.129   & 0.141         \\ \cline{2-8} 
                     & 3                          & 0.954      & 0.957   & 0.988         & 0.137      & 0.138   & 0.161         \\ \cline{2-8} 
                     & 4                          & 0.952      & 0.948   & 0.991         & 0.142      & 0.144   & 0.177         \\ \cline{2-8} 
                     & 5                          & 0.952      & 0.955   & 0.99         & 0.146      & 0.149   & 0.192         \\ \hline
\end{tabular}
\caption{Coverage probability and Length of confidence intervals for niche overlap measures in two multivariate normal distributions with equal mean and unequal covariance matrices. The nominal level is 95\%. The quantities $n$ and $m$ are the group sample sizes, Dim is the dimension $d$.  The method Bonf is based on the empirical bootstrap quantiles with a Bonferroni correction, MVT is based on the  equi-coordinate multivariate normal quantiles and Ellipse is based on the asymptotic elliptical confidence region. The simulation setting is described in Example \ref{ex1}.  The interval length is calculated as the mean of all component-wise interval lengths.} 
\label{table1}
\end{table}

\begin{example} \label{ex2}
In this example, we investigate the effect of correlation between the component variables.  Here also, we focus on the comparison of two \( d \)-dimensional multivariate normal distributions with equal mean vectors \( \bm \mu^{(g)} = \mathbf{0} \in \mathbb{R}^d \) and equal covariance matrices  \( \Sigma \) for $g\in \{1,2\}$. The matrix $\Sigma$ has constant off-diagonal entries equal to \( 0.75 \) and diagonal entries equal to \( 3 \).
This indicates fairly strong correlations between components ($\rho=0.25$). The true value of the overlap index is ${\rm I}(\bm F,\bm G)=0.5\cdot \bm 1_d$.

We observe similar trends as in Example \ref{ex1} and notice that higher correlation ($\rho = 0.75$) impacts coverage probabilities slightly but does not drastically change performance trends (Table \ref{table2}).
\end{example}
\begin{table}[htb]
\centering
\begin{tabular}{|c|c|c|c|c|c|c|c|}
\hline
\multirow{2}{*}{n=m} & \multirow{2}{*}{Dim.} & \multicolumn{3}{c|}{Coverage Probability} & \multicolumn{3}{c|}{Interval Length} \\ \cline{3-8} 
                             &                            & Bonf. & Mvt. & Ellip. & Bonf. & Mvt. & Ellip. \\ \hline
\multirow{4}{*}{50}          & 2                          & 0.965      & 0.955   & 0.969         & 0.268      & 0.271   & 0.295         \\ \cline{2-8} 
                             & 3                          & 0.964      & 0.956   & 0.982         & 0.286      & 0.289   & 0.337         \\ \cline{2-8} 
                             & 4                          & 0.973      & 0.961   & 0.988         & 0.297      & 0.302   & 0.371         \\ \cline{2-8} 
                             & 5                          & 0.978      & 0.972   & 0.994         & 0.306      & 0.312   & 0.400         \\ \hline
\multirow{4}{*}{100}         & 2                          & 0.966      & 0.957   & 0.971         & 0.186      & 0.187   & 0.205         \\ \cline{2-8} 
                             & 3                          & 0.968      & 0.961   & 0.987         & 0.198      & 0.200   & 0.234         \\ \cline{2-8} 
                             & 4                          & 0.960      & 0.960   & 0.993         & 0.206      & 0.209   & 0.257         \\ \cline{2-8} 
                             & 5                          & 0.963      & 0.962   & 0.994         & 0.212      & 0.215   & 0.278         \\ \hline
\multirow{4}{*}{200}         & 2                          & 0.943      & 0.942   & 0.963         & 0.130      & 0.131   & 0.143         \\ \cline{2-8} 
                             & 3                          & 0.943      & 0.949   & 0.979         & 0.139      & 0.140   & 0.163         \\ \cline{2-8} 
                             & 4                          & 0.956      & 0.956   & 0.993         & 0.144      & 0.146   & 0.178         \\ \cline{2-8} 
                             & 5                          & 0.946      & 0.951   & 0.989         & 0.149      & 0.150   & 0.194         \\ \hline
\end{tabular}
\caption{Coverage probability and Length of confidence intervals for niche overlap measures in two multivariate normal distributions with equal mean and  covariance matrices but high correlation among the component variables. The nominal level is 95\%. The quantities $n$ and $m$ are the group sample sizes, Dim is the dimension $d$.  The method Bonf is based on the empirical bootstrap quantiles with a Bonferroni correction, MVT is based on the  equi-coordinate multivariate normal quantiles and Ellipse is based on the asymptotic elliptical confidence region. The simulation setting is described in Example \ref{ex2}.  The interval length is calculated as the mean of all component-wise interval lengths.}
\label{table2}
\end{table}

\begin{example} \label{ex3}
To evaluate performance under heavy tails, we consider \( d \)-dimensional multivariate \( t \)-distributions with one degree of freedom (\( \nu = 1 \)), where the two distributions have the same mean vector \( \bm \mu^{(g)} = \mathbf{0} \in \mathbb{R}^d \) and equal scale matrices \( \Sigma \).  The entries of $\Sigma$ are the same as in Example \ref{ex2}.  Here also,  the true value of the overlap index is ${\rm I}(\bm F,\bm G)=0.5\cdot \bm 1_d$.

Again, overall the performances are similar as in Examples \ref{ex1} and \ref{ex2}, while the heavy-tailed nature of the t-distribution has only a mild effect in that the intervals are slightly conservative, especially for smaller sample sizes (Table \ref{table3}). 
\end{example}
\begin{table}[htb]
\centering
\begin{tabular}{|c|c|c|c|c|c|c|c|}
\hline
\multirow{2}{*}{n=m} & \multirow{2}{*}{Dim.} & \multicolumn{3}{c|}{Coverage Probability} & \multicolumn{3}{c|}{Interval Length} \\ \cline{3-8} 
                     &                            & Bonf. & Mvt. & Ellip. & Bonf. & Mvt. & Ellip. \\ \hline
\multirow{4}{*}{50}  & 2                          & 0.970      & 0.946   & 0.944         & 0.268      & 0.269   & 0.289         \\ \cline{2-8} 
                     & 3                          & 0.975      & 0.965   & 0.954         & 0.285      & 0.286   & 0.325         \\ \cline{2-8} 
                     & 4                          & 0.982      & 0.968   & 0.959        & 0.297      & 0.299   & 0.353         \\ \cline{2-8} 
                     & 5                          & 0.983      & 0.969   & 0.966         & 0.305      & 0.308   & 0.379         \\ \hline
\multirow{4}{*}{100} & 2                          & 0.958      & 0.951   & 0.930         & 0.186      & 0.186   & 0.200         \\ \cline{2-8} 
                     & 3                          & 0.963      & 0.961   & 0.930         & 0.198      & 0.198   & 0.224         \\ \cline{2-8} 
                     & 4                          & 0.964     & 0.963   & 0.950         & 0.206      & 0.206   & 0.259         \\ \cline{2-8} 
                     & 5                          & 0.969      & 0.961   & 0.949         & 0.212      & 0.212   & 0.259         \\ \hline
\multirow{4}{*}{200} & 2                          & 0.958      & 0.952   & 0.939         & 0.130      & 0.130   & 0.139         \\ \cline{2-8} 
                     & 3                          & 0.950      & 0.941   & 0.914         & 0.139      & 0.138   & 0.155         \\ \cline{2-8} 
                     & 4                          & 0.962      & 0.953   & 0.956         & 0.144      & 0.144   & 0.169         \\ \cline{2-8} 
                     & 5                          & 0.954      & 0.951   & 0.949         & 0.149      & 0.147   & 0.180         \\ \hline
\end{tabular}
\caption{Coverage probability and Length of confidence intervals for niche overlap measures in two multivariate $t$ distributions with equal mean and  covariance matrices. The nominal level is 95\%. The quantities $n$ and $m$ are the group sample sizes, Dim is the dimension $d$.  The method Bonf is based on the empirical bootstrap quantiles with a Bonferroni correction, MVT is based on the  equi-coordinate multivariate normal quantiles and Ellipse is based on the asymptotic elliptical confidence region. The simulation setting is described in Example \ref{ex3}.  The interval length is calculated as the mean of all component-wise interval lengths.}
\label{table3}
\end{table}

\subsection{Hypothesis Testing}
This section evaluates the empirical size and power of the four two-sample test procedures introduced in Section \ref{sec2.4}: (i) the Wald type test (Wald), (ii) the ANOVA-type test (ANOVA), (iii) the max-type test based on the maximum of univariate statistics (Max T), and (iv) the percentile test based on empirical bootstrap quantiles with Bonferroni correction (Percentile).
\begin{example} \label{ex4}
{\bf (Empirical Size)} We compare two \( d \)-dimensional multivariate normal distributions with equal mean vectors \( \bm \mu^{(g)} = \mathbf{1} \in \mathbb{R}^d \) and common covariance matrix \( \Sigma \), with unit variances and off-diagonal entries of \( 0.25 \). Simulations were conducted for \( d = 2 \) and \( d = 5 \).

As shown in Tables \ref{tab:rejection_rates1} and \ref{tab:rejection_rates_updated}, Wald and ANOVA tests maintain empirical sizes close to the nominal level, with the Wald test showing a tendency to be liberal in small samples and higher dimensions. Max T and Percentile tests are consistently conservative. The extremely small empirical sizes reflect slow
convergence of extremal bootstrap quantiles in dependent multivariate settings, rather
than a failure of the underlying asymptotic theory.
\end{example}  

\begin{table}[htb]
\centering
\begin{tabular}{cccccc}
\toprule
\textbf{Sample Size} & \textbf{Wald} & \textbf{ANOVA} & \textbf{Max T} & \textbf{Percentile} \\
\midrule
10  & 0.063 & 0.052 & 0.000 & 0.001 \\
20  & 0.037 & 0.034 & 0.000 & 0.000 \\
30  & 0.041 & 0.035 & 0.000 & 0.000 \\
40  & 0.050 & 0.048 & 0.000 & 0.000 \\
50  & 0.055 & 0.050 & 0.000 & 0.000 \\
60  & 0.048 & 0.039 & 0.001 & 0.000 \\
70  & 0.048 & 0.044 & 0.000 & 0.000 \\
80  & 0.052 & 0.050 & 0.000 & 0.000 \\
90  & 0.049 & 0.047 & 0.000 & 0.000 \\
100 & 0.054 & 0.048 & 0.004 & 0.000 \\
\bottomrule
\end{tabular}
\caption{Empirical sizes for the test introduced in Section \ref{sec2.4} for equal sample sizes $n=m$. The  Wald test is based on the Wald type statistic, the ANOVA Test on the ANOVA-type statistic, the max t test on the maximum of all univariate test statistics and the percentile test on the empirical bootstrap quantiles with a Bonferroni correction.  The simulation setting is described in Example \ref{ex4} for  dimension \(d=2\).}
\label{tab:rejection_rates1}
\end{table}

\begin{table}[htb]
\centering
\begin{tabular}{cccccc}
\toprule
\textbf{Sample Size} & \textbf{Wald} & \textbf{ANOVA} & \textbf{Max T} & \textbf{Percentile} \\
\midrule
10  & 0.079 & 0.035 & 0.000 & 0.000 \\
20  & 0.057 & 0.023 & 0.000 & 0.000 \\
30  & 0.064 & 0.038 & 0.000 & 0.000 \\
40  & 0.050 & 0.027 & 0.001 & 0.000 \\
50  & 0.046 & 0.038 & 0.000 & 0.000 \\
60  & 0.067 & 0.046 & 0.000 & 0.000 \\
70  & 0.044 & 0.033 & 0.000 & 0.000 \\
80  & 0.052 & 0.040 & 0.001 & 0.000 \\
90  & 0.048 & 0.042 & 0.001 & 0.000 \\
100 & 0.048 & 0.040 & 0.001 & 0.000 \\
\bottomrule
\end{tabular}
\caption{Empirical sizes for the test introduced in Section \ref{sec2.4} for equal sample sizes $n=m$. The  Wald test is based on the Wald type statistic, the ANOVA Test on the ANOVA-type statistic, the max t test on the maximum of all univariate test statistics and the percentile test on the empirical bootstrap quantiles with a Bonferroni correction.  The simulation setting is described in Example \ref{ex4} for dimension \(d=5\).}
\label{tab:rejection_rates_updated}
\end{table}

\begin{example}  \label{ex5}
{\bf (Power under Heterogeneous Variance)} Here, we consider distributions with mean vector \( \bm \mu^{(g)} = \mathbf{1} \in \mathbb{R}^d \).  We set the variances in  group 1 to \( \sigma^{(1)} = 1 \), and investigate the power performance of the tests as the variance in group 2,   \( \sigma^{(2)} \), increases from $1.1$ to $2.0$. Sample sizes are set to \( n = m = 100 \) and dimensions \( d = 2 \) and \( d = 5 \) are considered.

It is clear from Figure \ref{fig:power_sigma}  that the powers of  Wald and ANOVA tests  increase  as either of the shift in the variance  or dimensionality increase. Max T and Percentile tests remain less sensitive and require larger sample sizes to detect differences.
\end{example}
\begin{figure}[htb]
    \centering
    \includegraphics[width=0.8\textwidth, page=1]{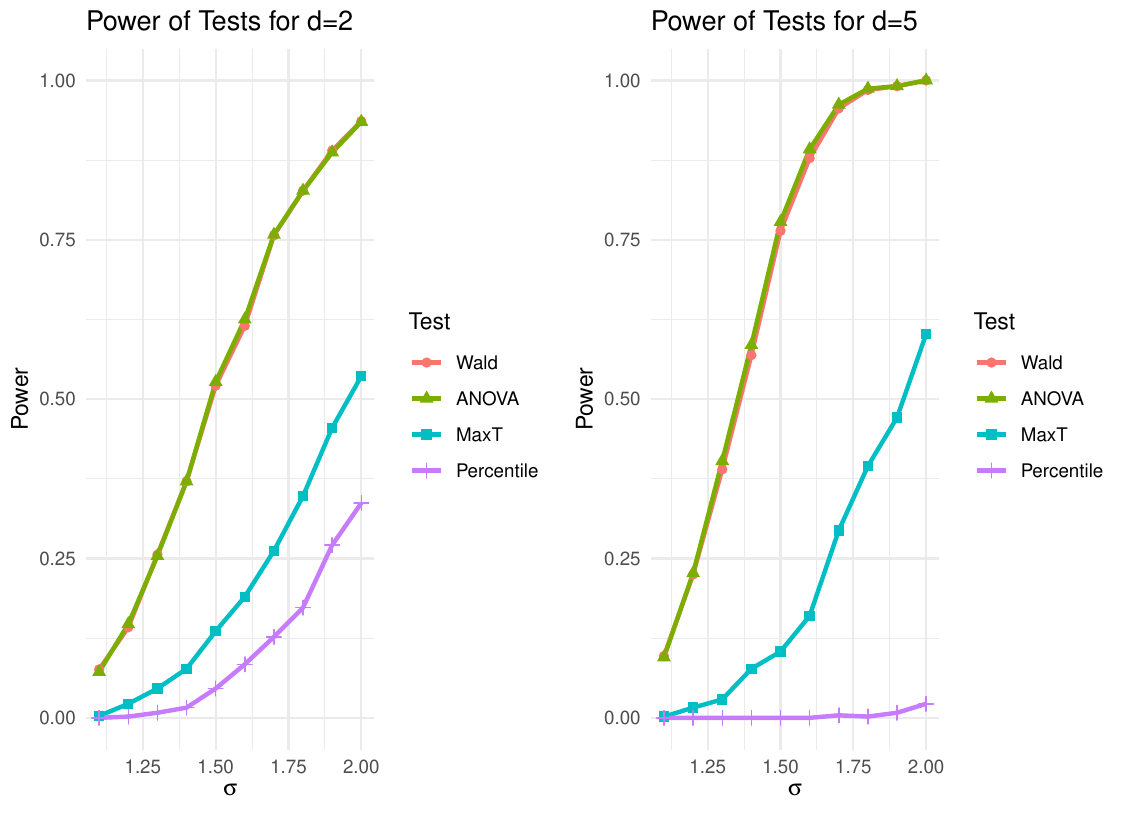}
    \caption{Empirical power for the tests from Section \ref{sec2.4}, based on $n=m=100$. The x-axis shows the standard deviation in group 2. Results are shown for dimensions $d=2$ and $d=5$, as described in Example \ref{ex5}.}
    \label{fig:power_sigma}
\end{figure}

\begin{example} \label{ex6}
{\bf (Power by Sample Size)} To study the effect of sample size on power, we fix the variance constant with \( \sigma^{(1)} = 1 \) and \( \sigma^{(2)} = 2 \) for all components, and vary the sample size. 

As shown in Figure \ref{fig:power_sample_size}, power of the Wald and ANOVA tests increases with sample size and is generally higher in higher dimensions. The Max T and Percentile tests remain overly conservative with limited power gains.
\end{example}
% Figure 2: ex_4.6.pdf
\begin{figure}[htb]
    \centering
    \includegraphics[width=0.8\textwidth, page=1]{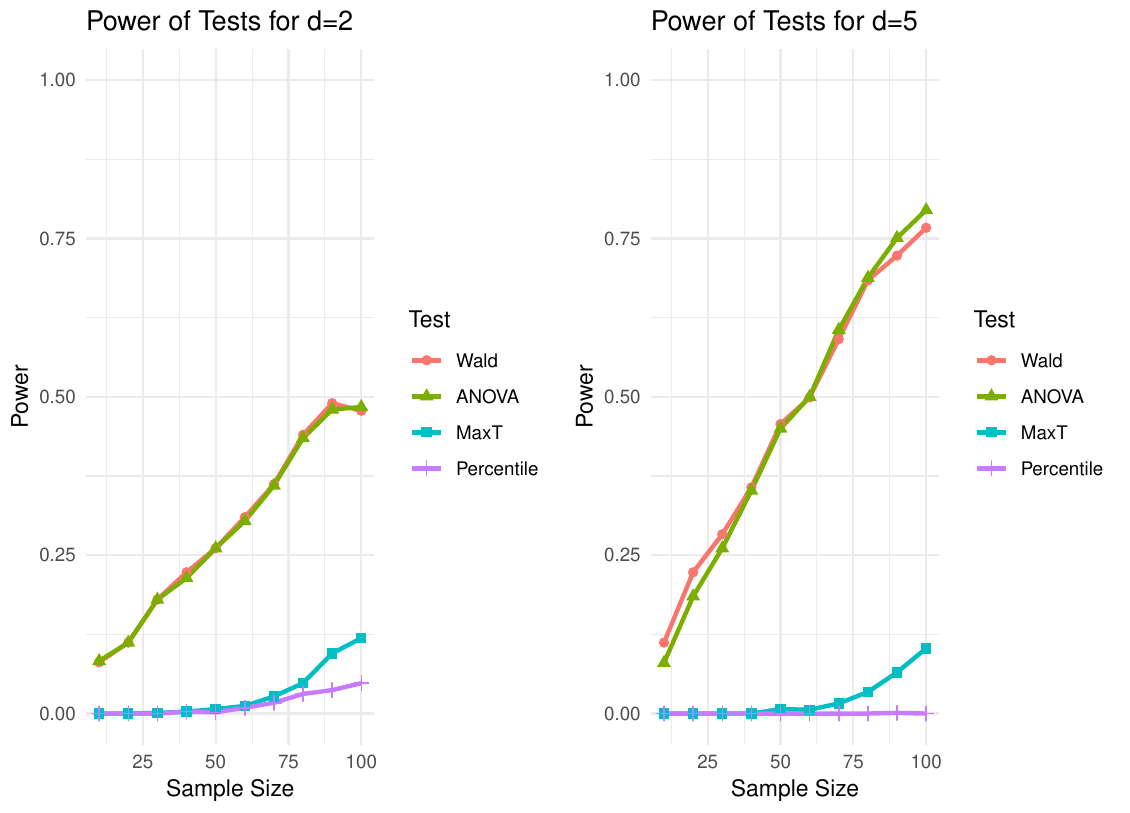}
    \caption{Empirical power for the tests from Section \ref{sec2.4}, for increasing sample sizes and fixed variances \( \sigma^{(1)} = 1 \), \( \sigma^{(2)} = 2 \). Results are shown for $d=2$ and $d=5$, as described in Example \ref{ex6}.}
    \label{fig:power_sample_size}
\end{figure}

\section{Proofs} \label{s3}
\subsection{Proof of Theorem \ref{th:asymptotic} } \label{A1}
There exists sequences of two independent Brownian bridges, $B_{1,n}^{(s)}(u)$ and $B_{2,n}^{(s)}(u)$,  $u \in [0,1]$, such that under Assumptions \ref{assum_1} and \ref{assum_2},
\begin{align} \label{thm_22}
& \sqrt{m}( \hat{F}^{(s)}_{n}((\hat{G}^{(s)}_{m})^{-1}(u))- F^{(s)}((G^{(s)})^{-1}(u)))  \\
&=  \sqrt{\nu} \ B_{1,N}^{(s)}(F^{(s)}((G^{(s)})^{-1}(u))) + \frac{f^{(s)}((G^{(s)})^{-1}(u))}{g^{(s)}((G^{(s)})^{-1}(u))}   B_{2,N}^{(s)}(u) \notag \\
&\quad + o((N)^{-1/2}(\log N)^2) \notag
\end{align}
almost surely for all $s=1, \ldots,d$, and this holds uniformly over any subinterval $[a,b]$ of $[0,1]$ \cite[][Theorem 2.2]{hsieh_turnbull}. Since $a$ and $b$ are selected without restriction, the relationship remains valid over all such subintervals.

Let $\mathbb{D}=D[0,1]$ be the Skorokhod space, i.e., the space of all C\`adl\`ag functions on $[0,1]$, and $\mathbb{E}\subset \mathbb{R}$. Consider the map $\phi: \mathbb{D} \rightarrow \mathbb{E}$ defined by \citep{ParkinsonKutilKupplerJunkerTrutschnigBathke+2018, beck2023combining} 
\begin{align} \label{phi}
 \phi(f^{(s)})=  \int_0^1 f^{(s)} \biggl(1- \frac{\alpha}{2} \biggl )d \alpha- \int_0^1 f^{(s)} \biggl(\frac{\alpha}{2} \biggl) d \alpha.
\end{align}   
Clearly, $\phi$ is a continuous linear map and  $$\phi(\hat{F}^{(s)} \circ (\hat{G}^{(s)})^{-1}) = I(\hat{F}^{(s)}, \hat{G}^{(s)}).$$ 
To generalize this to the multivariate situation, define now $\Phi: \mathbb{D}^d\to \mathbb{R}^d$ by
\begin{align} \label{Phi}
  \Phi (f^{(1)}, \ldots, f^{(d)} ) = (\phi(f^{(1)}), \ldots,\phi( f^{(d)} ))
\end{align}
and 
\begin{align*}
   \widehat{\mathcal{I}}_{n,m}(\mathbf{F},\mathbf{G})=\Phi(\hat{F}^{(1)} \circ (\hat{G}^{(1)})^{-1}, \ldots,\hat{F}^{(d)} \circ (\hat{G}^{(d)})^{-1}) 
\end{align*}

Since $\phi$ is continuously Hadamard-differentiable \citep[][Theorem 20.9]{vaart_1998}, so is $\Phi$ by chain rule. Then, the desired result follows by the Cr\`amer-Wold device and the delta method for empirical processes \citep[][Theorem 3.9.4]{van1996weak}. %
\subsection{Proof of Theorem \ref{thmboot} } \label{A2} 
  The conditional central limit theorem holds in outer probability for each bootstrapped empirical distribution function,
\begin{align*}
    \sqrt{n}(\hat{F}^*_n-\hat{F}_n) \ \ \ \text{and} \ \ \ \sqrt{m}(\hat{G}^*_m-\hat{G}_m),
\end{align*}
by Theorem 3.6.1 in \cite{van1996weak}. 

 We consider now the map
\[
\psi(F, G) = F \circ G^{-1}
\]
for distribution functions \( F \) and \( G \) defined on \( \mathbb{R} \). This transformation is Hadamard-differentiable tangentially to \( D[a, b] \times C[a, b] \)  \citep[][Comment 4 in Section 3.10]{vaart2023statistical}. Then the composition $\phi \circ \psi$, where $\phi$ is defined as in \eqref{phi}, is Hadamard-differentiable \citep[][Theorem 20.9]{vaart_1998} by chain rule.
  Again considering the multivariate map $\Phi$ as defined in \eqref{Phi}, the composition of these three maps are again Hadamard-differentiable.
 The desired result follows then directly by applying the delta method for empirical bootstrap processes \citep[][Theorem 3.9.11]{van1996weak} and the Cr\`amer-Wold device.
 \section{Additional Information on the Data Example} \label{s4}
For our motivating fish data example we provide here additional information on the structure of the data set. First we show a pairwise scatterplot matrix, where the diagonal cells show density plots for each group. Each upper-triangle cell contains the  Pearson correlation coefficient. Each lower-triangle cell is a scatter plot comparing two variables. 
\begin{figure}[h]
    \centering
    \includegraphics[width=\textwidth]{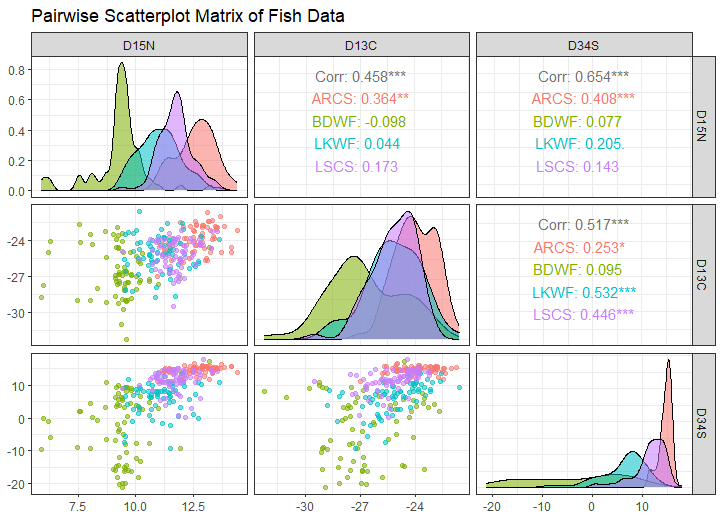}
    \caption{Pairwise Scatterplot Matrix of Fish Data}
    \label{fig:scatterplot_fish}
\end{figure}
The qqplots show that the assumption of multivariate normailty can not be justified.
\begin{figure}[h]
    \centering
    \includegraphics[width=\textwidth]{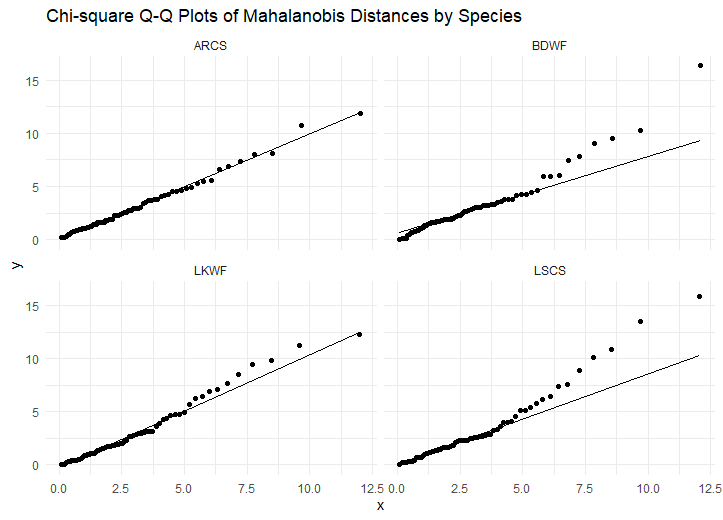}
    \caption{Chi-square Q-Q Plots of Mahalanobis Distances by Species}
    \label{fig:qqplot_fish}
\end{figure}

\end{document}